\documentclass[12pt]{article}

\usepackage{graphicx}
\usepackage{float}
\usepackage{amsmath}
\usepackage{amssymb}
\usepackage{amsthm}
\usepackage{comment}
\usepackage{xcolor}
\usepackage{hyperref}
\hypersetup{colorlinks = false}
\usepackage[numbers,sort&compress]{natbib}
\usepackage[margin=0.75in]{geometry}

\newcommand{\be}{\begin{equation}}
\newcommand{\ee}{\end{equation}}
\newcommand{\ba}{\begin{array}}
\newcommand{\ea}{\end{array}}
\newcommand{\bea}{\begin{eqnarray}}
\newcommand{\eea}{\end{eqnarray}}

\DeclareMathOperator{\tr}{Tr}
\newcommand{\ra}{\rangle}
\newcommand{\la}{\langle}

\newcommand{\calH}{{\cal H }}

\newcommand{\calE}{{\cal E }}

\newcommand{\calS}{{\cal S }}

\newcommand{\calZ}{{\cal Z }}

\newcommand{\CC}{\mathbb{C}}
\newcommand{\DD}{\mathbb{D}}
\newcommand{\RR}{\mathbb{R}}

\newcommand{\ket}[1]{|#1\rangle}
\newcommand{\bra}[1]{\langle #1|}

\newcommand{\eps}{\epsilon}
\newcommand{\bits}[1]{\{0,1\}^{#1}}

\newcommand{\odisk}[1]{{\mathbb{D}_{{#1}}}}

\newcommand{\multi}[1]{\mathbf{{#1}}}
\newcommand{\rr}{\multi{r}}

\newtheorem{dfn}{Definition}

\newtheorem{lemma}{Lemma}
\newtheorem{fact}{Fact}

\newtheorem{corol}{Corollary}
\newtheorem{conj}{Conjecture}
\newtheorem{theorem}{Theorem}

\newtheorem{numerics}{Numerical Study}

\newcommand{\footremember}[2]{%
    \footnote{#2}
    \newcounter{#1}
    \setcounter{#1}{\value{footnote}}%
}
\newcommand{\footrecall}[1]{%
    \footnotemark[\value{#1}]%
}

\title{Lee-Yang tensors and Hamiltonian complexity}

\author{Benjamin Wong\footremember{iqc}{Department of Combinatorics and Optimization and Institute for Quantum Computing, University of Waterloo}
\and
Sergey Bravyi\footremember{ibm}{IBM Quantum, IBM T.J. Watson Research Center, Yorktown Heights NY 10598}
\and
David Gosset\footrecall{iqc} \footremember{perimeter}{Perimeter Institute for Theoretical Physics}
\and
Yinchen Liu\footrecall{iqc} 
}

\begin{document}
\maketitle
\begin{abstract}

A complex tensor with $n$ binary indices can be identified with a multilinear polynomial in $n$ complex variables. We say it is a \textit{Lee-Yang tensor} with radius $r$ if the polynomial is nonzero whenever all variables lie in the open disk of radius $r$. In this work we study quantum states and observables which are Lee-Yang tensors when expressed in the computational basis.  

 We first review their basic properties, including closure under tensor contraction and certain quantum operations. We show that quantum states with Lee-Yang radius $r>1$ can be prepared by quasipolynomial-sized circuits. We also show that every Hermitian operator with Lee-Yang radius $r>1$ has a unique principal eigenvector. These results suggest that $r=1$ is a key threshold for quantum states and observables. 

Finally, we consider  a family of two-local Hamiltonians where every interaction term
energetically favors a deformed EPR state $|00\rangle+s|11\rangle$ for some $0\leq s\leq 1$. 
We numerically investigate this model and find that on all graphs considered the Lee-Yang radius  of the ground state is at least $r=1/\sqrt{s}$
while the spectral gap between the  two smallest eigenvalues is at least $1-s^2$. We conjecture that these lower bounds hold more generally; in particular, this would provide an efficient quantum adiabatic algorithm for the quantum Max-Cut problem on uniformly weighted bipartite graphs.

\end{abstract}

\section{Motivation\label{sec:motivation}}

Tensor networks are widely used to model many-body quantum systems. They provide a succinct description  of quantum states living in an exponentially large Hilbert space and serve as a powerful computational tool. By imposing a proper algebraic or geometric structure on the underlying tensors, one can often engineer states with desired physical properties, such as topological order~\cite{orus2019tensor,aguado2008entanglement,schuch2012topological,csahinouglu2021characterizing}, criticality~\cite{vidal2008class}, or even toy models of quantum gravity~\cite{pastawski2015holographic}. The structure of the tensors also plays a central role in the computational complexity of tensor network contraction. For example, tensor networks associated with stabilizer states~\cite{aaronson2004improved} or with matchgate tensors on planar graphs~\cite{knill2001fermionic,terhal2002classical,bravyi2009contraction} can be contracted in polynomial time, whereas contracting networks built from generic, unstructured tensors is \#P-hard~\cite{schuch2007computational,haferkamp2020contracting}.

In this work we revisit the class of Lee-Yang tensors, which has been studied extensively in the mathematics literature but has received comparatively little attention in the quantum computing community. These tensors exhibit a striking algebraic structure characterized by the zero sets of certain multivariate polynomials, as discussed below. Most importantly, the class of Lee-Yang tensors is closed under tensor products and under contraction of indices. 
This closure property is reminiscent of exactly solvable models in many-body physics. 
As a consequence, the corresponding family of $n$-qubit operators forms a semigroup, and the associated family of $n$-qubit states is closed under postselected Pauli measurements (in the $X$ or $Y$ basis). These features point to a rich underlying mathematical structure and suggest potential algorithmic implications that deserve further investigation.

Our starting point is the seminal 1952 paper of Lee and Yang \cite{lee1952}
who investigated the partition function of the ferromagnetic Ising model in the presence of a uniform complex-valued magnetic field $\mu\in \mathbb{C}$, establishing that all of its zeros lie on the imaginary axis $\mathrm{Re}(\mu)=0$ in the complex plane \cite{lee1952}. A powerful generalization due to Suzuki and Fisher describes a family of two-local \textit{quantum} Hamiltonians with this property \cite{sf71}. 

To make things concrete, let us consider the spin-$1/2$ Heisenberg model on a graph $G=(V,E)$, which describes $n=|V|$ qubits that interact via the Hamiltonian
\begin{equation}
H=\sum_{\{i,j\}\in E} J_{ij}(X_i X_j+Y_i Y_j+Z_i Z_j)
\label{eq:heis}
\end{equation}
Here $X,Y,Z$ are the Pauli matrices and $\{J_{ij}\}$ are arbitrary nonnegative coefficients. The problem of estimating the ground energy of this Hamiltonian has 
also been called ``quantum Max-Cut"; it is known that computing the ground energy to inverse polynomial precision is QMA-hard~\cite{piddock2015complexity}, and moreover it remains NP-hard even if we only ask for an  approximation to within a certain constant factor \cite{piddock2025quantum}. It is therefore extremely unlikely that efficient approximation algorithms exist for the ground energy of this Hamiltonian in the worst case.

Now let us consider the special case of  Eq.~\eqref{eq:heis} where the graph $G$ is bipartite. This restriction could in principle make the quantum Max-Cut problem easier. Indeed, whether or not the ground energy of Eq.~\eqref{eq:heis} for bipartite $G$ can be efficiently approximated is an open question \cite{gharibian2024guest} that has attracted considerable attention in the last few years~\cite{king2023improved, marwaha2024performance, tao2025refined, ju2025improved,apte20250,gribling2025improved}. 
If $V=AB$ is a bipartition of $G$ and we conjugate Eq.~\eqref{eq:heis} by $Y_A=\bigotimes_{i\in A} Y_i$ then we get a Hamiltonian
\[
H'=\sum_{\{i,j\}\in E} J_{ij} (-X_iX_j+Y_iY_j-Z_iZ_j)
\]
composed of two-qubit interactions 
$-X\otimes X + Y\otimes Y - Z\otimes Z$ which energetically favour the EPR state $(|00\rangle+|11\rangle)/\sqrt{2}$.
The partition function of this system\footnote{For simplicity, here
the partition function is evaluated at the inverse temperature $\beta=1$. This does not restrict the generality since one can absorb $\beta$
in the coefficients $J_{ij}$.}  in the presence
of a complex-valued magnetic field $\mu\in \CC$
is 
$\calZ(\mu)=\mathrm{Tr}\left(e^{-H'-\mu \sum_{i=1}^{n} Z_i}\right)$.
Suzuki and Fisher's extension of the Lee-Yang theorem \cite{sf71} shows that $\calZ(\mu)\ne 0$
whenever $\mathrm{Re}(\mu)\neq 0$.
 Their proof also elucidates a very special feature of the Gibbs state of this model, establishing that it is a \textit{Lee-Yang tensor} (with radius $1$). This means that the multivariate polynomial
\[
f(z_1,z_2,\ldots, z_{2n})=\sum_{x,y\in \{0,1\}^n} \langle x|e^{-H'}|y\rangle z_1^{x_1}\ldots z_n^{x_n} z_{n+1}^{y_1}\ldots z_{2n}^{y_n}
\]
with complex variables $z_1,z_2,\ldots,z_{2n}\in \CC$
does not vanish whenever $z_1,z_2,\ldots, z_{2n}$ are all inside the unit circle in the complex plane.

The result of Ref.~\cite{sf71} goes well beyond this example, and shows that this curious structural feature of the Gibbs state is shared by a broad class of quantum many-body systems, see Section \ref{sec:qmb} for details. A beautiful paper of Harrow, Mehraban, and Soleimanifar \cite{harrow2020classical} investigated whether this kind of zero-freeness guarantee can be leveraged to develop fast classical algorithms based on polynomial interpolation techniques. They ask:

\vspace{0.5cm}
\textit{Are quantum many-body systems that satisfy a Lee-Yang theorem computationally easy? }
\vspace{0.5cm}

Answering this question in the affirmative would of course resolve the complexity of the bipartite quantum Max-Cut problem. It would also extend the known families of quantum many-body systems which admit fast classical algorithms in arbitrary geometries; this includes the ferromagnetic Ising model \cite{jerrum1993polynomial} with a transverse magnetic field \cite{bravyi2014monte}, as well as the XY ferromagnet\cite{bravyi2016polynomial}.

In this work we pick up this thread and investigate Lee-Yang tensors through the lens of quantum complexity theory. Though we are not able to answer the above question, we establish several results concerning states and observables that are Lee-Yang tensors.

We note that Lee-Yang tensors, the associated multivariate stable polynomials and their generalizations have been studied fairly extensively in mathematics, see e.g. \cite{wagner2011multivariate} for a review. They have played a central role in celebrated results including the characterization of stability preserving operations due to Borcea and Br{\"a}nd{\'e}n \cite{borcea1,borcea2} and the resolution of the Kadison-Singer problem due to Marcus, Spielman, and Srivastava ~\cite{marcus2015interlacing}. In this work we focus specifically on the intersection of these topics with quantum information science.

This paper is structured as follows. In Section \ref{sec:lyt} we provide definitions of Lee-Yang tensors and review their basic properties including closure under tensor contraction and certain quantum operations. Then in Section \ref{sec:herm} we establish properties of Hermitian operators which are Lee-Yang tensors; among other things, we prove an analogue of the Perron-Frobenius theorem. We then consider the circuit complexity of quantum states which are Lee-Yang tensors with radius $r>1$; we show in Section \ref{sec:state-prep} that any such state can be prepared by a quasipolynomial-sized quantum circuit.

Finally, in Section \ref{sec:qmb} we return to the topic of quantum many-body systems. We review the work of Suzuki and Fisher which establishes that certain Gibbs states are Lee-Yang tensors with radius $r=1$. Then we consider a subclass of  Suzuki-Fisher
Hamiltonians which can be viewed as a deformation of the Heisenberg model Eq.~\eqref{eq:heis} when restricted to a bipartite graph. Each term in the Hamiltonian is proportional to a projector onto a deformed EPR state $|00\rangle+s|11\rangle$, where $0\leq s\leq 1$.  
We numerically investigate this model and find that on all graphs considered the Lee-Yang radius 
of the ground state is at least $1/\sqrt{s}$
while the spectral gap between the smallest and second smallest eigenvalues is at least $1-s^2$. We formulate a plausible conjecture regarding the Lee-Yang radius of the Gibbs state and discuss implications for the quantum Hamiltonian complexity theory. 
We conclude with a discussion of some open questions in Section~\ref{sec:open}.

\section{Lee-Yang tensors \label{sec:lyt}}
Suppose $\psi \in (\CC^2)^{\otimes n}$ is  a complex tensor with $n$ binary indices
and entries $\psi_x$ labelled by bit strings $x\in \{0,1\}^n$. 
We shall be interested in tensors describing multi-qubit state vectors, density matrices,
operators, and quantum channels. 
Define a generating polynomial $f_\psi \, : \, \CC^n \to \CC$ as 
\be\label{eq:lee-yang-poly}
f_\psi(z) = \sum_{x \in \{0,1\}^n} \psi_x \prod_{a\in \mathrm{supp}(x)} z_a.
\ee
Here $z=(z_1,\ldots,z_n)$ is a vector of complex variables and $\mathrm{supp}(x)$ denotes the support of a bit string $x$, that is,
the set of all indices $i\in \{1,2,\ldots,n\}$ such that $x_i=1$.
Equivalently, $f_\psi(z)$ is the inner product between an $n$-qubit state $|\psi\ra=\sum_{x\in \{0,1\}^n}\; \psi_x |x\ra$ and a tensor product of single-qubit states $|0\ra + z_a^*|1\ra$
with $a=1,\ldots,n$. In the special case $n=0$ the tensor $\psi \in \CC$ is a scalar and we set $f_\psi=\psi\in \CC$.
Suppose  $r>0$ is a real number. Let
\[
\odisk{r} = \{z \in \CC\, : \, |z|<r\} 
\]
be the open zero-centered disk of radius $r$.
 If $\rr =(r_1,\ldots,r_n)\in \RR_+^n$ is a tuple of radii, let
\[
\odisk{\rr} = \odisk{r_1} \times \cdots \times  \odisk{r_n}
\]
be the corresponding multi-disk in $\CC^n$.
\begin{dfn}
We say that $\psi \in (\CC^2)^{\otimes n}$ is a Lee-Yang tensor with 
a radius   $\rr \in \RR_+^n$ if the generating polynomial of $\psi$ has no zeros
in the open zero-centered  multi-disk  of radius $\rr$, that is,
\be
f_\psi(z)\ne 0 \quad \mbox{for all $z\in \odisk{\rr}$}.
\ee
A scalar $\psi\in\CC$ is a Lee-Yang tensor iff $\psi\ne 0$. 
\end{dfn}
Let $LY_n(\rr)$ be the set of all Lee-Yang tensors $\psi\in (\CC^2)^{\otimes n}$  with the  radius $\rr$.  
Given a real number $r>0$, let
$LY(r)$ be the set of  all Lee-Yang tensors  with  the radius at least $r$
for each variable:
\be
LY(r)=\bigcup_{n\ge 0} \; LY_n((r,r,\ldots,r)).
\ee

The following lemma is a rephrasing of results from~\cite{asano1970theorems,ruelle2010characterization}.
\begin{lemma}[\bf Tensor contraction]
\label{lemma:contraction}
Consider any tensor $\psi \in LY_n(\rr)$ with $n\ge 2$. Let $\phi\in (\CC^2)^{\otimes (n-2)}$ be a tensor obtained from $\psi$ by contracting some pair of indices $i<j$
such that $r_i r_j>1$. Then $\phi \in LY_{n-2}(\rr')$ where $\rr'$ is obtained from $\rr$ by deleting the components $r_i$ and $r_j$. 
Furthermore, if $r_i r_j=1$ then $\phi$ is either identically zero or $\phi \in LY_{n-2}(\rr')$.
\end{lemma}
For completeness, we provide a proof in the Appendix~\ref{app:A}. It is sometimes convenient to consider a sequence of tensors $\{\psi_j\}_{j\geq 1}$ such that $\psi_j\in LY(r)$ for each $j\geq 1$. If the sequence converges to a nonzero tensor $\psi$, then $\psi\in LY(r)$ as well. As we now explain, this is a simple consequence of Hurwitz's theorem (see, e.g., \cite{choe2004homogeneous}).

\begin{theorem}[\textbf{Hurwitz's Theorem}]
Suppose $n$ is a positive integer, $D$ is a domain in $\mathbb{C}^{n}$,  and $\{f_k\}_{k\in\mathbb{N}}$ is a sequence of non-vanishing analytic functions on $D$ that converges to $f$ uniformly on compact subsets of $D$. Then $f$ is either non-vanishing or identically zero on $D$.
\end{theorem}

\begin{corol}[\textbf{Limit of Lee-Yang tensors}]
Let $n\geq 1$ be a positive integer and let $\rr\in \mathbb{R}_{+}^n$. Consider a sequence of tensors $\psi_j \in (\mathbb{C}^{2})^{\otimes n}$ for $j\geq 1$. Suppose that for each $x\in \{0,1\}^n$ we have $\langle x|\psi_j \rangle \rightarrow \langle x|\psi\rangle$ as $j\rightarrow \infty$, for some $\psi \in (\mathbb{C}^{2})^{\otimes n}$. If $\psi_j\in LY_n(\rr)$ for each $j\geq 1$ and $\psi$ is not the zero vector, then $\psi\in LY_n(\rr)$.
\label{corol:seq}
\end{corol}
\begin{proof}
Let $\rr=(r_1,r_2,\ldots, r_n)$ and let $R=\max_{1\leq i\leq n} r_i$. Given $\epsilon>0$ we can choose $k\in \{1,2,\ldots\}$ such that $|\langle x|\psi_j\rangle -\langle x|\psi\rangle|\leq (2R)^{-n}\epsilon$ for all $x\in \{0,1\}^n$ and $j\geq k$. Then for $z\in \mathbb{D}_\rr$ and $j\geq k$ we have
\begin{align*}
|f_{\psi_j}(z)-f_{\psi}(z)|&\leq \sum_{x\in \{0,1\}^n} |\langle x|\psi_j\rangle -\langle x|\psi\rangle|\prod_{a\in \mathrm{supp}(x)}|z_a|\\
&\leq \sum_{x\in \{0,1\}^n} (2R)^{-n}\epsilon R^n\\
&=\epsilon,
\end{align*}
which shows that our sequence of polynomials $\{f_{\psi_j}\}_j$ converges uniformly in $\mathbb{D}_\rr$ and therefore also on compact subsets of it. Each $f_{\psi_j}$ is non-vanishing on $\mathbb{D}_\rr$. Since $f_{\psi}$ is not identically zero on $\mathbb{D}_\rr$, it follows from Hurwitz's theorem that $f_{\psi}$ is non-vanishing on $\mathbb{D}_\rr$, i.e.,  $\psi\in LY_{n}(\rr)$. 
\end{proof}

Lemma~\ref{lemma:contraction} shows in particular that any tensor network obtained by contracting tensors in $LY(r)$ with $r\geq 1$ outputs a tensor which is also in $LY(r)$. 
The lemma also implies that full-rank $n$-qubit operators contained in $LY(1)$ form  a semigroup.
Indeed, if $A,B\in LY(1)$ then $A\otimes B\in LY(1)$. Since $AB$ can be obtained from $A\otimes B$ by contracting
$n$ pairs of indices and $AB\ne 0$, Lemma~\ref{lemma:contraction} gives  $AB\in LY(1)$.

Obviously, a single-qubit state $|\psi\ra=\psi_0|0\ra + \psi_1 |1\ra$ is contained in $LY(1)$ iff $|\psi_0|\ge|\psi_1|$.
The following lemma gives an analogous condition for containment in $LY(1)$ for two-qubit states.
\begin{lemma}
\label{lemma:LY1}
A  two-qubit state  $\psi$ is contained in $LY(1)$ if and only if
\be
\label{condLY1}
|\psi_{10} \psi_{00}^* - \psi_{11}\psi_{01}^*| + |\psi_{11}\psi_{00} - \psi_{10}\psi_{01}| \le |\psi_{00}|^2 - |\psi_{01}|^2.
\ee
\end{lemma}
\begin{proof}
Let $a=\psi_{11}$, $b=\psi_{10}$, $c=\psi_{01}$, and $d=\psi_{00}$. The generating polynomial of $\psi$ is
$f_\psi(z)=cz_2 + d + z_1(az_2 + b)$. Hence $f_\psi(z)\ne 0$ for $z_1,z_2\in \DD_1$ iff
$|cz_2 + d|>|az_2 + b|$ for $z_2\in \DD_1$. This is equivalent to 
$M(\DD_1)\subseteq \DD_1$ for a M{\"o}bius map
\[
M(v) = \frac{a v + b}{c v + d}
\]
Theorem~1 from Ref.~\cite{martin2006composition} asserts that $M(\DD_1)\subseteq \DD_1$
iff $|b \bar{d} - a\bar{c}| + |ad - bc| \le |d|^2 - |c|^2$, which is equivalent to 
Eq.~(\ref{condLY1}).
\end{proof}

The following theorem gives an easy-to-verify sufficient condition for a tensor to be in $LY(1)$. The statement below is a slight extension of a result established in Ref.~\cite{sf71}, which corresponds to the $+$ sign in Eq.~\eqref{eq:sf1}.

\begin{theorem}[\textbf{Suzuki-Fisher 1971} \cite{sf71}]
Suppose $\psi\in (\mathbb{C}^2)^{\otimes m}$ satisfies
\begin{equation}
X^{\otimes m}|\psi\rangle=\pm |\psi^{*}\rangle
\label{eq:sf1}
\end{equation}
and
\begin{equation}
|\langle 0^{m}|\psi\rangle|\geq \frac{1}{4}\sum_{y\in \{0,1\}^{m}} |\langle y|\psi\rangle|
\label{eq:abs}
\end{equation}
Then $\psi\in LY(1)$.
\label{thm:sf71}
\end{theorem}
We provide a proof in Appendix~\ref{app:sf71}.

 A striking corollary of Lemma~\ref{lemma:contraction} is that the set of multi-qubit states contained in
 $LY(r)$ 
is closed under a class of  post-selective measurements that includes measurements in the Pauli $X$ or $Y$ basis,
provided that $r\ge 1$.
To describe this formally,
suppose $|\theta\ra$ is a single-qubit normalized state. 
Let $n$ be the total number of qubits.
A post-selective measurement  of the $i$-th qubit
is an operator $M_\theta\, : \, (\CC^2)^{\otimes n} \to (\CC^2)^{\otimes (n-1)}$ 
that applies the projector $|\theta\ra\la \theta|$ to the $i$-th qubit and then discards the $i$-th qubit.
The measurement is called equatorial if 
\be
\label{theta_state}
|\theta\ra=\frac1{\sqrt{2}}(|0\ra + e^{i\theta}|1\ra)
\ee
for some angle $\theta\in [0,2\pi)$. For example, measurements in the Pauli $X$ or $Y$ basis are equatorial.
\begin{lemma}
\label{lem:1}
Suppose $r\ge 1$ and $|\psi\ra\in LY(r)$ is an $n$-qubit state vector.
For any post-selective equatorial measurement  $M_\theta$ one has
$M_\theta |\psi\ra \in LY(r)$ or $M_\theta |\psi\ra=0$.
\end{lemma}
\begin{proof}
Let $|\theta^*\ra$ be the complex conjugate of $|\theta\ra$.
The generating polynomial of $|\theta^*\ra$ is $(1 + e^{-i\theta} z)/\sqrt{2}\ne 0$ for $|z|<1$.
Hence $|\theta^*\ra \in LY(1)$. Suppose $\psi \in LY_n((r,r,\ldots,r))$ and we measure the $i$-th qubit of $\psi$. We have
\[
|\psi\ra \otimes |\theta^*\ra \in LY_{n+1}((r,r,\ldots,r,1)).
\]
By definition, $M_\theta |\psi\ra$ is obtained from the tensor $|\psi\ra \otimes |\theta^*\ra$  by contracting
the $i$-th and the last index. 
Lemma~\ref{lemma:contraction} then implies that either $M_\theta |\psi\ra=0$ or $M_\theta |\psi\ra \in LY_{n-1}((r,r,\ldots,r))$. 
The case $M_\theta |\psi\ra=0$ can only occur if $r=1$. Hence $M_\theta$ maps $LY(r)$ to $LY(r)\cup \{0\}$.
\end{proof}

We can use Theorem \ref{thm:sf71} to construct single-qubit quantum channels that preserve the LY property. Let us first review some basic facts about quantum channels (see, e.g., \cite{watrous2018theory}). Recall that if $\mathcal{E}$ is a $1$-qubit quantum channel, its Choi matrix is a $4\times 4$ complex matrix defined by
\[
C(\mathcal{E})\equiv(\mathcal{E}\otimes I)|\Omega\rangle\langle \Omega| \qquad |\Omega\rangle=|00\rangle+|11\rangle.
\]
Moreover, a $4\times 4$ complex matrix is the Choi matrix of a single-qubit quantum channel iff it satisfies 
\[
C(\mathcal{E})\geq 0 \qquad \text{ and } \mathrm{Tr}_1 (C(\mathcal{E}))=I.
\]
Let us say that a channel $\mathcal{E}$ is in $LY(r)$ iff its Choi matrix satisfies $C(\mathcal{E})\in LY(r)$. Note that the output of a quantum channel $\mathcal{E}$ on input state $\rho$ can be computed by contracting $C(\mathcal{E})$ with $\rho$. In particular, suppose $\rho$ is an $n$-qubit state and $\mathcal{E}$ is a $1$-qubit channel. Let $C'(\mathcal{E})$ be the partial transpose of $C(\mathcal{E})$ on the first qubit. Then
\[
\mathrm{Tr}_2((C'(\mathcal{E})\otimes I_{n-1}) (I_1\otimes \rho))=(\mathcal{E}\otimes I_{n-1})(\rho)
\]
Note that  $C'(\mathcal{E})\in LY(1)$ iff $C(\mathcal{E})\in LY(1)$.  From this and Lemma \ref{lemma:contraction} we infer that $\mathcal{E}\in LY(1)$ and $\rho\in LY(1)$ imply that the output state $(\mathcal{E}\otimes I_{n-1})(\rho)\in LY(1)$.

Consider a single-qubit Pauli channel
\[
\mathcal{E}_{\vec{p}}(\rho)= p_0 \rho+ p_1 X\rho X+p_2 Y\rho Y+p_3 Z\rho Z.
\]
Here $\vec{p}$ is a probability distribution over $\{0,1,2,3\}$.

\begin{theorem}
 The set of multi-qubit density matrices contained in $LY(1)$ is closed under single-qubit channels $\mathcal{E}_{\vec{p}}$ for all $\vec{p}$ satisfying $\min\{p_0,p_3\}\geq \max\{p_1,p_2\}$.
\end{theorem}
\begin{proof}
 The Choi matrix  is
\[
C(\mathcal{E}_{\vec{p}})=
\begin{pmatrix}
p_0+p_3 & 0 & 0 & p_0-p_3\\
0 & p_1+p_2 & p_1-p_2  & 0\\
0 & p_1-p_2  & p_1+p_2  & 0\\
p_0-p_3 & 0 & 0& p_0+p_3
\end{pmatrix}
\]
It suffices to show that $C(\mathcal{E}_{\vec{p}})\in LY(1)$. To this end we can vectorize the Choi matrix (i.e., view it as a vector in $\mathbb{C}^{16}$) and verify Eqs.~(\ref{eq:sf1},\ref{eq:abs}) from Theorem \ref{thm:sf71}.  Eq.~\eqref{eq:sf1} is equivalent to $(X\otimes X)C(\mathcal{E}_{\vec{p}}) X\otimes X=\pm C(\mathcal{E}_{\vec{p}})^{*}$, which is satisfied with the $+$ sign.   

Eq.~\eqref{eq:abs} gives
\[
p_0+p_3 \geq \frac{1}{2}\left(p_0+p_1+p_2+p_3+|p_0-p_3|+|p_1-p_2|\right).
\]
Rearranging the above we get
\[
\left(p_0+p_3-|p_0-p_3|\right)\geq \left(p_1+p_2+|p_1-p_2|\right),
\]
or equivalently $\min\{p_0,p_3\}\geq \max\{p_1,p_2\}$.
\end{proof}

For example, the depolarizing channel with parameter $\lambda\in [0,1]$ corresponds to the case 
\[
p_0=(1-3\lambda/4) \quad p_1=p_2=p_3=\lambda/4,
\]
and we see that it is in $LY(1)$ for all $\lambda$. Similarly, the dephasing channel with parameter $q\in [0,1]$ corresponds to the choice $p_1=p_2=0$, $p_0=1-q$, $p_3=q$, which again is in $LY(1)$ for all $q$.

\section{Hermitian operators \label{sec:herm}}

Next let us consider 
Hermitian operators (observables) contained in $LY(1)$. 
For example, one can easily check that the identity $I$, Pauli $Z$ operator, 
SWAP gate, and an operator $|0^n\ra\la 0^n|-|x\ra\la x|$ with any nonzero $x\in \{0,1\}^n$
are contained in $LY(1)$. The same holds for any tensor product of such operators since $LY(1)$ is closed under tensor products.

The following lemma can be viewed as a generalization of the Griffiths inequality~\cite{griffiths1967correlations} asserting that the expected value of any
$Z$-type Pauli observable on the Gibbs state of the ferromagnetic  Ising  model is non-negative. 
\begin{lemma}[\bf Griffiths inequality]
\label{lem:2}
Suppose $P\in LY(1)$ is a Hermitian  $n$-qubit  operator such that $\la 0^n|P|0^n\ra>0$.
Then for any $n$-qubit density matrix $\rho \in LY(1)$ one has 
$\mathrm{Tr}(\rho P) \ge 0$.
\end{lemma}
\begin{proof}
Let $D_s = |0\ra\la 0| + s|1\ra\la 1|$.
Consider a family of density matrices
\[
\rho_s = D_s^{\otimes n} \rho D_s^{\otimes n}, \qquad 0\le s\le 1.
\]
From $\rho \in LY(1)$ one gets $\rho_s \in LY(1/s)$.
Applying Lemma~\ref{lemma:contraction} to perform the contraction
$P\otimes \rho_s\to \mathrm{Tr}(\rho_s P)$ one concludes that 
\be
\label{Griffiths_eq1}
\mathrm{Tr}(\rho_s P)\ne 0 \quad \mbox{for all $s\in [0,1)$}.
\ee
We have  $\la 0^n|\rho|0^n\ra>0$ since $\la 0^n|\rho|0^n\ra\ge 0$ and
$\la 0^n|\rho|0^n\ra\ne 0$ due to the inclusion $\rho \in LY(1)$.
By assumption, $\la 0^n|P|0^n\ra>0$. Thus
\be
\label{Griffiths_eq2}
\mathrm{Tr}(\rho_0 P) = \la 0^n|\rho|0^n\ra \la 0^n|P|0^n\ra>0.
\ee
From  Eqs.~(\ref{Griffiths_eq1},\ref{Griffiths_eq2})
and continuity of the function $s\to \mathrm{Tr}(\rho_s P)$ one gets $\mathrm{Tr}(\rho P)=\mathrm{Tr}(\rho_1 P) \ge 0$.
\end{proof}

Next we show that {\em real-valued} tensors contained in $LY(r)$ for some $r>1$
are sign-problem-free in the sense that all their entries must have the same sign,
up to a simple  basis change. 
\begin{lemma}
\label{lem:3}
Suppose $P$ is an $n$-qubit operator  with real matrix elements in the standard basis.
Suppose $P\in LY(r)$ for some $r>1$
and $\la 0^n|P|0^n\ra>0$.
 Then $P$ has  strictly positive matrix elements in the Pauli $X$  basis.
\end{lemma}
\begin{proof}
Let $H$ be the Hadamard gate. One can easily check that  $H|0\ra\in LY(1)$ and $H|1\ra \in LY(1)$.
Hence $H^{\otimes n}|x\ra \in LY(1)$ for all $x\in \{0,1\}^n$.
Applying Lemma~\ref{lemma:contraction} to perform the contraction
\[
P \otimes H^{\otimes n} |x\ra \to PH^{\otimes n}|x\ra
\]
and using $r>1$
gives $PH^{\otimes n} |x\ra\in LY(r)$. Choose any $y\in \{0,1\}^n$.
Applying Lemma~\ref{lemma:contraction} again to perform the contraction
\[
\la y|H^{\otimes n} \otimes PH^{\otimes n}|x\ra \to \la y|H^{\otimes n} P H^{\otimes n} |x\ra
\]
and using $r>1$ and $H^{\otimes n}|y\ra \in LY(1)$ 
gives 
\be
\label{projector_eq1}
\la y|H^{\otimes n} P H^{\otimes n}|x\ra \ne 0
\ee
 for all $x,y\in \{0,1\}^n$.
 Accordingly, it suffices to consider two cases.
 
\noindent
{\em Case~1:} All matrix elements $\la y|H^{\otimes n} P H^{\otimes n}|x\ra$ are strictly positive. 
Then we are done.

\noindent
{\em Case~2:} At least one matrix element $\la y|H^{\otimes n} P H^{\otimes n}|x\ra$ is negative.
Consider a function 
\[
g(s) = \la y|H^{\otimes n} D_s^{\otimes n} P D_s^{\otimes n} H^{\otimes n} |x\ra, \qquad D_s= |0\ra\la 0| + s|1\ra\la 1|
\]
with $s\in [0,1]$. From $P\in LY(r)$ one gets
$D_s^{\otimes n} P D_s^{\otimes n} \in LY(r/s)$. Hence the same argument as above shows that $g(s)\ne 0$
for all  $s\in [0,1]$.
Since $g(1)<0$, continuity argument gives $g(0)<0$. However, 
$g(0)=2^{-n} \la 0^n|P|0^n\ra >0$
which is a contradiction.
\end{proof}
We note that the proof of Lemma~\ref{lem:3} can be also applied  to operators $P$ that map $m$ qubits to $n$ qubits with $m\ne n$.
Choosing $m=0$  shows that any $n$-qubit state vector $\psi\in LY(r)$ with $r>1$,  real amplitudes in the standard basis, and
$\la 0^n|\psi\ra>0$ must have positive amplitudes in the Pauli-$X$ basis.
If the operator $P$ of Lemma~\ref{lem:3} is Hermitian, the Perron-Frobenius theorem implies that the maximum  eigenvalue of $P$
is simple (has multiplicity one). We now show that any complex-valued Hermitian operator contained in $LY(r)$ for some $r>1$
must have a simple dominant eigenvalue.

\begin{theorem}
\label{thm:degeneracy}
Suppose  $r>1$ and $H\in LY(r)$ is an $n$-qubit Hermitian operator. 
Then the eigenvalue of $H$ with largest magnitude has multiplicity $1$.
\end{theorem}
\begin{proof}
Suppose  $r>1$ and $H\in LY(r)$ is an $n$-qubit Hermitian operator.  Suppose the eigenvalue of $H$ with largest magnitude has multiplicity $d\in \{1,2,\ldots\}$. Consider $P\equiv H^2$. $P$ is positive semidefinite and its maximal eigenvalue has multiplicity that is at least $d$. Moreover, $P\in LY(r)$ using Lemma \ref{lemma:contraction} and the fact that $H\in LY(r)$. To complete the proof, in the following we show that the multiplicity of the largest eigenvalue of $P$ cannot be larger than $1$.

Let $\lambda$ be the maximum eigenvalue of $P$. From $P\in LY(r)$ one gets $P\ne 0$ and $\lambda>0$.  Define
\[
\Pi_m = (P/\lambda)^m  \quad \mbox{and} \quad \Pi = \lim_{m\to \infty} \Pi_m.
\]
Clearly, $\Pi$ is a projector and the rank of $\Pi$ coincides with the multiplicity of $\lambda$.
Thus we need to show that $\Pi$ cannot have rank greater than 1. First we claim that 
\be
\label{degeneracy_eq0}
\Pi \in LY(r).
\ee

Indeed, $\Pi$ is the limit of a sequence of tensors  $(P/\lambda)^m$ for integers $m\ge 1$. By Lemma~\ref{lemma:contraction}, each tensor in this sequence is in $LY(r)$. By Corollary \ref{corol:seq} and the fact that $\Pi\neq 0$, we get Eq.~(\ref{degeneracy_eq0}).

\begin{lemma}
Suppose $r>1$. There is no $n$-qubit projector $\Pi\in LY(r)$ with rank $\geq 2$.
\end{lemma}
\begin{proof}
Suppose to reach a contradiction that $\Pi\in LY(r)$ has rank $k\geq 2$. Write
\begin{equation}
\Pi=\sum_{j=1}^{k} |\psi_j\rangle\langle \psi_j|.
\label{eq:pi}
\end{equation}
where $\langle \psi_i|\psi_j\rangle=\delta_{ij}$. The set of normalized vectors in the image of $\Pi$ can be identified with the set of unit vectors in $\mathbb{C}^k$ which we denote 
\[
S=\{u\in \mathbb{C}^{k}: \|u\|=1\}.
\]
Define a set
\[
T=\{u\in S: \sum_{j=1}^{k} u_j |\psi_j\rangle \in LY(r)\}.
\]
Below we show that the set $T$ has the following properties:
\begin{enumerate}
\item $T$ is compact.
\item $T$ does not contain two orthogonal vectors.
\item For any vector $v\in S\setminus T$, there is a vector $w\in T$ such that $\langle v|w\rangle=0$.
\end{enumerate}

Clearly $T$ is a bounded set, so to show it is compact it suffices to show that it is closed. This latter property follows from Corollary\ref{corol:seq}, which shows that a limit of $LY(r)$ tensors is either the zero tensor or it is in $LY(r)$ (since no sequence of points in $S$ has the zero vector as a limit). To show property 2., suppose to reach a contradiction that $v,w\in T$ satisfy $\langle v|w\rangle=0$. Then letting 
\[
|\tilde{v}\rangle=\sum_{j=1}^{k} v_j |\psi_j\rangle \qquad \text{and} \qquad |\tilde{w}\rangle=\sum_{j=1}^{k} w_j|\psi_j\rangle,
\]
we have $\langle \tilde{v}|\tilde{w}\rangle=0$. But $|\tilde{v}\rangle,|\tilde{w}\rangle\in LY(r)$ with $r>1$ (since $v,w\in T$) and so their inner product is the contraction of two $LY(r)$ tensors which cannot be zero by Lemma \ref{lemma:contraction}.

Finally, let us show Property 3. To this end, suppose  $v\in S\setminus T$. Without loss of generality let us suppose $v=(1,0,0,\ldots, 0)$ (if it is not we can simply redefine our choice of orthonormal basis $\psi_1, \psi_2,\ldots, \psi_k$ in Eq.~\eqref{eq:pi}. Then  $|\psi_1\rangle \notin LY(r)$, so there exists $\vec{z}\in \mathbb{D}_r^n$ such that $\langle \vec{z}|\psi_1\rangle=0$.  But then $\Pi |\vec{z}^{\star}\rangle \in LY(r)$, where
\[
\Pi|\vec{z}^{\star}\rangle=\sum_{j=2}^{k} |\psi_j\rangle\langle \psi_j|\vec{z}^{\star}\rangle.
\]
Letting 
\[
w=\frac{1}{\sqrt{\sum_{j=2}^{k} |\langle \psi_j|\vec{z}^{\star}\rangle|^2}}(0, \langle \psi_2|\vec{z}^{\star}\rangle, \langle \psi_3|\vec{z}^{\star}\rangle,\ldots,  \langle \psi_k|\vec{z}^{\star}\rangle)
\]
we see that $w\in T$ and $\langle v|w\rangle=0$. We have thus established properties 1-3. above.

To complete the proof of the theorem, we show that there is no set $T\subseteq S$ satisfying properties 1-3. Indeed, suppose to reach a contradiction that $T$ is such a set. From property 3 we infer that $T$ cannot be empty.  Let $f:T\times T\rightarrow [0,\pi/2]$ be the function 
\[
f(u,v)=\mathrm{arccos}(|\langle u|v\rangle|)
\]
This function defines a distance on the unit sphere and in particular it satisfies the triangle inequality (see, e.g., Eq.~(9.86) from Ref.~\cite{nielsen2010quantum}):
\begin{equation}
f(u,v)\leq f(v,w)+f(w,u) \qquad u,v,w\in S.
\label{eq:triangle}
\end{equation}
The function $f$ is also continuous and it therefore achieves its maximum value in the compact set $T\times T$. This maximum value $\alpha\equiv \max_{u,v\in T} f(u,v)$ satisfies $\alpha\in [0,\pi/2)$ by Property 2.

Now let $x\in T$ be an arbitrary vector and let $y\in S$ be any vector such that $\langle y|x\rangle=0$. Define
\[
r(\theta)=\cos(\theta)|x\rangle +\sin(\theta)|y\rangle \qquad \theta\in [0,\pi/2].
\] 
Note that $r(0)=x\in T$ and $r(\pi/2)=y\notin T$. Since $T$ is a compact set and $r$ is a continuous function, the image $I\equiv r([0,\pi/2])\cap T$ is also a compact set. Consider the function $g:I\rightarrow [0,\pi/2]$ defined by $g(r(\theta))=\theta$. Since $g$ is continuous and $I$ is compact, this function achieves its maximum value in $I$, which is strictly less than $\pi/2$ since $r(\pi/2)\notin T$. Therefore there is some point $\theta_0\in [0,\pi/2)$ such that
\begin{equation}
r(\theta_0)\in T \qquad\text{ and } \qquad r(\phi)\notin T \quad  \phi\in (\theta_0, \pi/2].
\label{eq:rtheta0}
\end{equation}
Define a set 
\[
W(\theta)=\{ v\in S: \langle v|r(\theta)\rangle=0\} \quad \theta \in [0,\pi/2].
\]

Now consider the value $\beta=\theta_0+\epsilon$ for some $\epsilon\in (0,\pi/2-\theta_0]$ to be chosen later. From Eq.~\eqref{eq:rtheta0} we have 
\begin{equation}
r(\beta)\in S\setminus T
\label{eq:rb}
\end{equation}
Moreover, by the triangle inequality Eq.~\eqref{eq:triangle}
\begin{align}
\min_{u\in W(\beta)} f(u, r(\theta_0))&\geq \min_{u\in W(\beta)} \left( f(u, r(\beta))-f(r(\beta),r(\theta_0))\right)\\
&=\pi/2-\epsilon.
\end{align}
Let us now choose $\epsilon$ to be a small positive constant such that $\pi/2-\epsilon>\alpha$. With this choice we have
\[
\min_{u\in W(\beta)} f(u, r(\theta_0))>\alpha=\max_{c,d\in T} f(c,d).
\]
Since $r(\theta_0)\in T$ we infer that
\begin{equation}
W(\beta)\cap T=\emptyset.
\label{eq:wb}
\end{equation}
Now Eqs.~(\ref{eq:rb},\ref{eq:wb}) states that all vectors in $S$ that are orthogonal to $r(\beta)\in S\setminus T$ are not in $T$. This contradicts Property 3, completing the proof.
\end{proof}
\end{proof}
Finally, let us describe a large family of Hermitian operators in $LY(1)$.
\begin{lemma}
\label{lemma:schur}
Let $|z\ra = \bigotimes_{a=1}^n (|0\ra + z_a^*|1\ra)$ and define an $n$-qubit Hermitian operator
\[
P = |z\ra\la z| - X^{\otimes n}|z^*\ra\la z^*|X^{\otimes n}.
\]
Then $P \in LY(1)$ whenever $z\in \DD_1^n$. 
\end{lemma}
\begin{proof}
Consider any complex vectors $u,v\in \DD_1^n$. 
Since $\la u|P|v\ra = f_P(v^*,u)$, it suffices to prove that $\la u|P|v\ra\ne 0$.
We have
\begin{align}
\la u|P|v\ra &  = \prod_{a=1}^n (1+u_a z_a^*)(1+v_a^* z_a) -  \prod_{a=1}^n (u_a + z_a)(v_a^* + z_a^*)\nonumber \\
& =\prod_{a=1}^n (1+u_a z_a^*)(1+v_a^* z_a)  \left( 1 - \prod_{a=1}^n \frac{(u_a + z_a)(v_a^* + z_a^*)}{(1+u_a z_a^*)(1+v_a^* z_a)} \right).
\label{Gamma_element}
\end{align}
We claim that 
\be
\label{upper_a}
\left|  \frac{(u_a + z_a)(v_a^* + z_a^*)}{(1+u_a z_a^*)(1+v_a^* z_a)}\right|<1
\ee
for all $a$. Indeed, we have
\[
|1+u_a z_a^*|^2 - |u_a + z_a|^2 = 1 + |u_a|^2 |z_a|^2 + 2\mathrm{Re}(u_a z_a^*) - |u_a|^2 - |z_a|^2 - 2\mathrm{Re}(u_a z_a^*) = (1-|z_a|^2)(1-|u_a|^2) >0.
\]
Thus
\[
\left|\frac{u_a + z_a}{1+u_a z_a^*}\right|<1.
\]
The same argument shows that 
\[
\left|\frac{v_a^* + z_a^*}{1+v_a^* z_a}\right|<1
\]
which proves Eq.~(\ref{upper_a}). Substituting Eq.~(\ref{upper_a}) into Eq.~(\ref{Gamma_element}) shows that 
$\la u|P|v\ra\ne 0$ for all $u,v\in \DD_1^n$, that is, $P \in LY(1)$. 
\end{proof}
\section{Circuits for Lee-Yang tensors \label{sec:state-prep}}

In this section we show that any $n$-qubit quantum state which is a Lee-Yang tensor with radius $r>1$ is (quasi)-efficiently describable in the sense that (a) there exists a quasipolynomial-sized classical circuit that, given any $y\in \{0,1\}^n$, approximates the corresponding $X$-basis amplitude $\langle y|H^{\otimes n}| \psi\rangle$ of the state  to a given relative error, and  (b) there exists a quasipolynomial-sized quantum circuit that approximately prepares the state, to within a given trace distance.
An important caveat is that the classical and quantum circuits depend on the state $\psi$. 
We also assume that $r>1$ is a fixed constant independent of $n$.

The proof of these results is based on Barvinok's polynomial interpolation lemma~\cite{barvinok2016book}, which is a general technique for approximating polynomials which have a zero-free region. In this Section we say that $\tilde{z}\in\mathbb{C}$ is an \textit{$\eps$-relative error estimate} for $z\in\mathbb{C}$ iff $z=\tilde{z}(1+w)$ where $|w|\leq \eps$. 
 
\begin{lemma}[\textbf{Barvinok polynomial interpolation}~\cite{barvinok2016book}]
    Let $f(z)$ be a complex polynomial of degree $n$ satisfying
    \begin{equation}
       f(z)\neq 0 \text{ for all }|z|\leq r
    \end{equation}
    for some $r>1.$ Consider $g(z)=\ln(f(z))$. If $T_{p}(z)$ is the degree $p$ Taylor polynomial of $g(z)$ centered at $z=0$, then we have
    \begin{equation}
        |g(z)-T_p(z)|\leq\frac{n}{(p+1)r^p(r-1)}\text{ for all }|z|\leq 1.
    \end{equation}
    \label{lem:barvinok}
\end{lemma}

Lemma~\ref{lem:barvinok} implies that we can choose $p=O(\log(n/\eps))$ so that $\exp{\left(T_p(1)\right)}$ is an $\eps$-relative error estimate for $f(1)$. Moreover, this estimate can be computed using derivatives of $g(z)$ at $z=0$ up to order $p$.  When $f(0)$ and the coefficients of the polynomial $f$ are known up to degree $p$, these derivatives can be computed by solving a triangular system of linear equations.   In particular, letting $f^{(j)}, g^{(j)}$ be the $j$th derivatives of $f,g$ respectively,  evaluated at $z=0$, we have
\begin{equation}
f^{(m)}=\sum_{j=0}^{m-1} {m-1 \choose j} f^{(j)}g^{(m-j)} \qquad 1\leq m\leq p,
\label{eq:tri}
\end{equation}
see Ref.~\cite{barvinok2016book} for details.

The following lemma uses Barvinok's polynomial  interpolation to prove that the Lee-Yang radius
is robust against  small perturbations of a tensor. As a consequence,  amplitudes of the target state $\psi \in LY(r)$ to be prepared by a quantum circuit only need to be specified 
with $poly(n)$ bits of precision. 
\begin{lemma}
\label{lemma:robustness}
Suppose $r>1$ and let $\rho=(1+r)/2$. Let $\psi\in LY(r)$ be a normalized $n$-qubit state
and $\phi$ be an arbitrary $n$-qubit state
such that 
\be
\|\psi - \phi\| \le 2^{-n/2} (1+\rho^2)^{-n/2}  \exp{\left[  -n \frac{(r+1)}{(r-1)}\right]}.
\ee
Then $\phi \in LY(\rho)$.
\end{lemma}
\begin{proof}
Let $\epsilon\equiv \|\psi-\phi\|$.
Pick any $z\in \CC^n$ such that $|z_j|\le \rho$ for all $j$.
The triangle inequality and Cauchy–Schwarz  give
\be
|f_\phi(z)-f_\psi(z)| \le \sum_{x\in \{0,1\}^n} \; |\phi_x-\psi_x| \rho^{|x|} \le  \epsilon (1+\rho^2)^{n/2}.
\ee
Thus
\be
\label{f_phi(z)_lower}
|f_\phi(z)| \ge |f_\psi(z)| - \epsilon (1+\rho^2)^{n/2}\ge \delta -  \epsilon (1+\rho^2)^{n/2},
\ee
where
\be
\label{delta_def}
\delta = \min_{|z_1|,\ldots,|z_n|\le \rho} \; \; |f_\psi(z)|.
\ee
Let us derive a lower bound on $\delta$. Let $z$
be any point such that the minimum in Eq.~(\ref{delta_def}) is attained at $z$.
Consider a univariate polynomial 
$g(t) = f_\psi(tz)$, where  $t\in \CC$.
 Clearly, $g(t)$ has degree at most $n$ and $g(t)$
has no zeros if $|t|<R\equiv r/\rho$. Applying the polynomial interpolation lemma (Lemma~\ref{lem:barvinok}) with the order $p=0$ one gets
\be
|\log{g(1)} -\log{g(0)}|\le \frac{n}{R-1}.
\ee 
Taking the real part gives
\be
|\mathrm{Re}(\log{g(1)}) - \mathrm{Re}(\log{g(0)}) |\le \frac{n}{R-1}.
\ee
For any non-zero complex number $c$ one has $\mathrm{Re}(\log{c})=\log{|c|}$.
Since $|g(1)|=|f_\psi(z)|=\delta$ and $|g(0)|=|f_\psi(0)| = |\la 0^n|\psi\ra|$, 
one gets
\be
|\log{\delta} - \log{|\la 0^n|\psi\ra|} \, |\le  \frac{n}{R-1}.
\ee
Thus $\log{\delta} \ge  \log{|\la 0^n|\psi\ra|} - n/(R-1)$. Exponentiating this inequality gives
\be
\delta \ge |\la 0^n|\psi\ra| \exp{\left[ -\frac{n}{R-1}\right]}.
\ee
Griffiths inequality (Lemma~\ref{lem:2}) with the observable $P=|0^n\ra\la 0^n|-|x\ra\la x|\in LY(1)$ 
gives $|\la 0^n|\psi\ra|\ge |\la x|\psi\ra|$ for all $x\in \{0,1\}^n$. Since $\psi$ is normalized, this implies
$|\la 0^n|\psi\ra|\ge 2^{-n/2}$, that is,
\be
\delta \ge 2^{-n/2} \exp{\left[ -\frac{n}{R-1}\right]}.
\ee
Substituting this into Eq.~(\ref{f_phi(z)_lower}) gives
\be
|f_\phi(z)| \ge 2^{-n/2} \exp{\left[ -\frac{n}{R-1}\right]}
- \epsilon (1+\rho^2)^{n/2}.
\ee
Thus $f_\phi(z)\ne 0$ if $|z_j|\le \rho$ for all $j$  provided that 
\be
\epsilon <2^{-n/2} (1+\rho^2)^{-n/2}  \exp{\left[ -\frac{n}{(r/\rho)-1}\right]}.
\ee
This is equivalent to the statement of the lemma if one substitutes $\rho=(1+r)/2$.
\end{proof}
For simplicity, below we assume that amplitudes of the target state $\psi \in LY(r)$ to be prepared by a quantum circuit  are specified exactly. 
However, if amplitudes of $\psi$ are only specified with $poly(n)$ bits of precision, one can apply Lemma~\ref{lemma:robustness}
to replace $\psi$ with the approximating state at the cost of reducing the Lee-Yang radius from $r$ to $r'=(r+1)/2$.
Clearly, $r'>1$ whenever $r>1$.

The following theorem shows that $X$-basis amplitudes of states in $LY(r)$ with $r>1$ can be computed in quasipolynomial time. We note that the classical circuit described below depends only on the entries $\langle x|\psi\rangle$ with Hamming weight $|x|=O(\log(n/\eps))$.

\begin{theorem}[\textbf{Classical estimation of amplitudes}]
Let $\ket{\psi}\in(\mathbb{C}^2)^{\otimes n}$ be in $LY(r)$ for $r>1$ and $\eps>0.$ There is a classical circuit of size $n^{O(\log(n/\eps))}$ which takes as input $y\in \{0,1\}^n$ and outputs an $\eps$-relative error estimate of the corresponding $X$-basis amplitude
\[
\langle y|H^{\otimes n} |\psi\rangle.
\]
\label{thm:est}
\end{theorem}
\begin{proof}
    For any $y\in \{0,1\}^n$, consider the degree-$n$ univariate polynomial 
    \begin{equation}
        F_y(z)=\big(\bigotimes_{i=1}^n\frac{1}{\sqrt{2}}(\bra{0}+z(-1)^{y_i}\bra{1})\big)\ket{\psi}
        \label{eq:barv_polynomial}
    \end{equation}
Then $F_y(z)\neq 0$ when $|z|<r,$ by Lemma~\ref{lem:1}. The coefficients of $F_y$ are 
\begin{equation}
    F_y(z)=\sum_{j=0}^{n} c_{y,j} z^j \qquad c_{y,j}=\frac{1}{\sqrt{2^n}}\sum_{\substack{x\in\bits{n}:\\|x|=j}}\langle x|\psi\rangle (-1)^{x\cdot y}.
\end{equation}
Let $p=C\log(n/\eps)$ for $C=O(1)$ and note that for each $0\leq j\leq p$ we can compute the function $y\rightarrow c_{y,j}$ in quasipolynomial time using the expression above (this computation also involves the entries $\langle x|\psi\rangle$ for $|x|\leq p$ which can be hardcoded into our circuit). Hence we can compute the coefficients of $F_y$ up to order $p$ in time $n^{O(\log(n/\eps))}$. Also note that $F_y(0)=2^{-n/2} \langle 0^n|\psi\rangle$ is a constant independent of $y$. From this data we can then compute the derivatives of $\ln F_y(z)$ up to order $p$ using the triangular linear system Eq.~\eqref{eq:tri}, as discussed above.
Using Lemma~\ref{lem:barvinok}, we obtain an $\eps$-relative error estimate for 
\begin{equation}
    F_y(1)=\bra{y}H^{\otimes n}\ket{\psi}
\end{equation}
as desired.
\end{proof}

Next we show how to prepare $\ket{\psi}\in LY(r)$ with $r>1$ in quasipolynomial time on a quantum computer. We shall use the well-known Grover-Rudolph state preparation technique as a subroutine~\cite{grover2002creatingsuperpositionscorrespondefficiently}. This is a method of preparing a superposition
\begin{equation}
\sum_{x\in \{0,1\}^n} \sqrt{p(x)}|x\rangle
\label{eq:qsamp}
\end{equation}
corresponding to a classical probability distribution $p(x)$; states of this form where all entries are nonnegative real numbers are sometimes called ``quantum samples". The method requires the ability to compute marginal probabilities of the form
\[
p_j(y)=\sum_{z\in \{0,1\}^{n-j}}p(yz)
\]
where $1\leq j\leq n$ and $y\in \{0,1\}^j$. To prepare the state Eq.~\eqref{eq:qsamp} one proceeds in a qubit-by-qubit manner, starting from the single-qubit state
\[
\sqrt{p_1(0)}|0\rangle+\sqrt{p_1(1)}|1\rangle.
\]
One then adjoins a second qubit in the $|0\rangle$ state and performs a single qubit rotation conditioned on the state of the first qubit, to prepare:
\[
\sum_{y,z\in \{0,1\}} \sqrt{p_1(y)}|y\rangle|0\rangle \rightarrow \sum_{y,z\in \{0,1\}} \sqrt{p_1(y)}|y\rangle \sqrt{\frac{p_2(yz)}{p_1(y)}}|z\rangle=\sum_{y\in \{0,1\}^2} \sqrt{p_2(y)}|y\rangle.
\]
Continuing in this way, after $n-1$ conditional single-qubit rotations, one obtains the desired state Eq.~\eqref{eq:qsamp}. This procedure can be made to work even if the estimates of marginal probabilities $p_j(y)$ are only accurate to within a given relative error. In that case the runtime and accuracy of the method are summarized as follows.

\begin{lemma}[\textbf{Grover-Rudolph}~\cite{grover2002creatingsuperpositionscorrespondefficiently}]\label{lemma:grstateprep} Let $p(x)$ be a probability distribution on $\{0,1\}^n$ and $\eps > 0.$ Consider the quantum sample
\begin{equation}
    |\psi_p\rangle=\sum_{x\in\{0,1\}^n}\sqrt{p(x)}|x\rangle.
\end{equation}
Suppose that for each $1\leq j\leq n$ there is an algorithm with runtime $T$ which takes as input $y\in \{0,1\}^{j-1}$ and computes a corresponding marginal probability
\[
p_j(y)=\sum_{z\in \{0,1\}^{n-j}}p(yz)
\]
to within a given $O(\eps^{2}/n)$ relative error. Then there is a quantum algorithm with runtime $O(nT)$ which prepares a pure state $|\psi'_p\rangle$ satisfying
\begin{equation}
    \lVert |\psi'_p\rangle-|\psi_p\rangle \rVert\leq\eps.
\end{equation}    
\end{lemma}
The following theorem shows that states in $LY(r)$ with $r>1$ are prepared by quasipolynomial size quantum circuits.
\begin{theorem}[\textbf{Quantum state preparation}]
Let $\ket{\psi}\in(\mathbb{C}^2)^{\otimes n}$ be in $LY(r)$ for $r>1$ and $\eps>0.$ There is a quantum circuit of size $n^{O(\log(n/\eps))}$ which prepares an $n$-qubit state $|\psi'\rangle$ satisfying $\||\psi\rangle-|\psi'\rangle\|\leq \eps$.
\label{thm:sprep}
\end{theorem}
\begin{proof}
For each $0\leq k\leq n$ let
\begin{equation}
    \rho_k=\tr_{k+1,\ldots,n}(\ket{\psi}\bra{\psi})
\end{equation}
be the density matrix obtained from $|\psi\rangle$ by  tracing out the last $n-k$ qubits. The measured distribution of $ H^{\otimes n}\ket{\psi}$ is given by
    \begin{equation}
        p(x)=|\bra{x}H^{\otimes n}\ket{\psi}|^2.
    \end{equation}
We also consider marginal probabilities for $0\leq k\leq n$ defined by
    \begin{equation}
        p_k(y)=\bra{y}H^{\otimes k}\rho_k H^{\otimes k}\ket{y} \qquad y\in \{0,1\}^k.
        \end{equation}
The density matrix $\rho_k$ is in $LY(r)$ by Lemma~\ref{lemma:contraction} because it is obtained from the tensor $|\psi\rangle\langle \psi|$ by contracting indices corresponding to qubits $k+1,\ldots,n$. Note that we can regard $\rho_k$ as a vector with $4^k$ entries, so that each marginal $p_k(y)$ is an entry of the vector. We can  therefore use Theorem \ref{thm:est} to infer there is a classical circuit with runtime  $n^{O(\log(n/\delta))}$ which computes the marginals $y\rightarrow p_k(y)$ to within a given relative error $\delta$. Using this subroutine with $\delta=O(\eps^2/n)$ in Lemma~\ref{lemma:grstateprep} gives a quantum algorithm with runtime $n^{O(\log(n/\eps))}$ that prepares a state $\ket{\phi}$ that approximates the quantum sample corresponding to the distribution $p$:
    \begin{equation}
        \lVert\ket{\phi}-\sum_{x\in\bits{x}}\sqrt{p(x)}\ket{x}\rVert\leq\eps/2.
        \label{eq:approxamps}
    \end{equation}

Next we shall modify this state to ensure that each entry has the correct complex phase.  Observe that using Theorem~\ref{thm:est} to compute an $O(\eps)$-relative error approximation for $\bra{y} H^{\otimes n}\ket{\psi}$,  we can infer an $\eps/2$-additive error estimate $\tilde{\theta}_{y}$ of the phase $\theta_y=\arg(\bra{y} H^{\otimes n}\ket{\psi})$. This approximation to the phase can be computed by a reversible circuit which acts on $n+\log(\eps^{-1})$ qubits as
    \begin{equation}
        \ket{y}\ket{a}\mapsto\ket{y}\ket{a\oplus\tilde{\theta}_y}.
    \end{equation}
    With the ancilla register initialized to $\ket{0^{\log(\eps^{-1})}}$, apply a single-qubit diagonal unitary to each ancilla qubit to recover the phase
    \begin{equation}
        \ket{y}\ket{\tilde{\theta}_y}\mapsto e^{i\tilde{\theta}_y}\ket{y}\ket{\tilde{\theta}_y}.
    \end{equation}
    Next, uncompute and trace out the ancilla register, resulting in the state
    \begin{equation}
        \ket{\phi}\mapsto\sum_{x\in\bits{n}}\phi_xe^{i\tilde{\theta}_x}\ket{x}=:\ket{\phi'},
    \end{equation}
    which is within $\eps/2$ of the phase-approximate state 
    \begin{equation}
        \ket{\tilde\psi}=\sum_{x\in\bits{n}}\sqrt{p(x)}e^{i\tilde{\theta}_x}\ket{x}
    \end{equation}
    by Eq.~\eqref{eq:approxamps}. The state $\ket{\tilde{\psi}}$ is also within $\eps/2$ of $H^{\otimes n}\ket{\psi}$, as can be seen via
    \begin{equation}
        \begin{split}
            \|H^{\otimes n}\ket{\psi}-\ket{\tilde{\psi}}\|^2&=\sum_{x\in\bits{n}}p(x)|e^{i\theta_x}-e^{i\tilde{\theta}_x}|^2\\
        &=\sum_{x\in\bits{n}}p(x)(2-e^{i(\theta_x-\tilde{\theta}_x)}-e^{-i(\theta_x-\tilde{\theta}_x)})
        \end{split}
    \end{equation}
    and using the fact that $|\theta_x-\tilde\theta_x|<\eps/2$ for all $x\in\bits{n}.$ Hence by the triangle inequality we have
    \begin{equation}
        \|H^{\otimes n}\ket{\psi}-\ket{\phi'}\|\leq\eps
    \end{equation}
    and the state $\ket{\psi'}=H^{\otimes n}\ket{\phi'}$ satisfies $\|\ket{\psi}-\ket{\psi'}\|\leq \eps$ as desired.
\end{proof}

\section{Quantum many-body systems \label{sec:qmb}}

In this section we argue that thermal Gibbs states and ground states of many interesting quantum many-body systems
are Lee-Yang tensors with the radius $r=1$ and provide numerical evidence for $r>1$ in certain cases\footnote{These results also extend to canonical purifications of the Gibbs state in a straightforward way. Indeed, it follows from the definitions that the ``thermofield double" state,  proportional to $e^{-\beta H/2} \otimes I\sum_{z}|z\rangle|z\rangle$, has the same Lee-Yang radius as $e^{-\beta H/2}$.}.
We shall consider Hamiltonians that describe a system of $n$ qubits with two-qubit interactions.
The qubits are located at the vertices
of an interaction graph $G=(V,E)$ with a set of vertices  $V=\{1,\ldots,n\}$ such that a pair of qubits $(i,j)$ can interact  only if $(i,j)\in E$.
We write $X_i,Y_i,Z_i$ for the Pauli operators acting on the $i$-th qubit. 

\subsection{Remark on Lee-Yang radius $r=1$}
Before proceeding, let us briefly comment on the fact that in this section we often discuss quantum many-body states with Lee-Yang radius $r=1$ whereas we might prefer if the Lee-Yang radius were strictly greater than $1$ (for instance to make use of the results in Section \ref{sec:state-prep}). One way to deal with this is to use the following naive method  to increase the Lee-Yang radius of a given state. In particular, if $|\psi\rangle \in LY(1)$ is a normalized $n$-qubit state then let
\[
|\psi(r)\rangle=\frac{1}{\|D_{1/r}^{\otimes n}|\psi\rangle\|}D_{1/r}^{\otimes n}|\psi\rangle
\]
 where $D_{s}=\mathrm{diag}(1,s)$. By definition of the Lee Yang radius, we have $|\psi(r)\rangle \in LY(r)$. Although this transformation with $r>1$ expands the Lee-Yang radius, it also can drastically change the state;  unfortunately $|\psi(r)\rangle$ is not in general close to $|\psi\rangle $ in trace distance. However, the situation is actually not so bad if $|\psi\rangle$ is an eigenstate of a local Hamiltonian $H$ and if we measure closeness with respect to energy rather than trace distance. For example, suppose $H=\sum_{ij} h_{ij}$ is a two-local Hamiltonian such that each term $h_{ij}$ acts nontrivially only on qubits $i,j$, and that $H|\psi\rangle=\lambda|\psi\rangle$. In that case it is not hard to show that the expected value of $H$ on the state $|\psi(1+\mu)\rangle$ for a small constant $\mu=O(1)$ approximates $\lambda$ to within an additive term that is bounded as $O(\mu)\cdot \sum_{ij}\|h_{ij}\|$.

 To see this note that 
 \begin{equation}
 \langle \psi(r)| H'(r)|\psi(r)\rangle=\lambda,
 \label{eq:hprime}
 \end{equation}
 where
 \[
H'(r)=D_{1/r}^{\otimes n} H  (D_{1/r}^{-1})^{\otimes n}=\sum_{i,j} D_{1/r}^{\otimes 2} h_{ij}(D_{1/r}^{-1})^{\otimes 2} 
 \]
 is in general non-Hermitian. From Eq.~\eqref{eq:hprime} we have
 \[
 |\langle \psi(r)|H|\psi(r)\rangle -\lambda|= |\langle \psi(r)|H'(r)-H|\psi(r)\rangle|\leq \sum_{i,j}\|D_{1/r}^{\otimes 2} h_{ij}(D_{1/r}^{-1})^{\otimes 2}-h_{ij}\|.
 \]
 Finally, note that if $r=1+\mu$ for a small constant $\mu$, then the RHS is upper bounded as $O(\mu)\cdot \sum_{ij}\|h_{ij}\|$.

\subsection{Suzuki-Fisher models}
\label{subs:SF}

Consider an anisotropic Heisenberg model with external magnetic fields described  by a Hamiltonian
\be
\label{SFmodel1}
H=-\sum_{(i,j)\in E} H_{i,j} - \sum_{i\in V} H_i,
\ee
where
\be
\label{SFmodel2}
H_{i,j} =J^x_{i,j} X_i X_j + J^y_{i,j} Y_i Y_j + J^z_{i,j} Z_i Z_j 
\ee
and
\be
\label{SFmodel3}
H_i = \mu^x_i X_i + \mu^y_i Y_i + \mu^z_i Z_i.
\ee
Here $J$'s and $\mu$'s are real coefficients. 
\begin{dfn}
An $n$-qubit Hermitian operator $H$ is called a Suzuki-Fisher Hamiltonian 
if it has the form defined in Eqs.~(\ref{SFmodel1},\ref{SFmodel2},\ref{SFmodel3}) 
where all $Z_iZ_j$ terms  are ferromagnetic and dominating such that
\be
\label{ferro1}
J^z_{i,j} \ge \max\{ |J^x_{i,j}|, |J^y_{i,j}|\}
\ee
for all $(i,j)\in E$ and all magnetic $Z$-fields are non-negative such that $\mu^z_i\ge0$ for all $i\in V$.
Let $\mathsf{SF}(n)$ be the set of all $n$-qubit Suzuki-Fisher Hamiltonians. 
\end{dfn}
Suzuki and Fisher~\cite{sf71} proved the following. 
\begin{fact}
\label{fact:sf71}
The thermal Gibbs state $\rho =e^{-\beta H}/\mathrm{Tr}(e^{-\beta H})$ of any Hamiltonian $H\in \mathsf{SF}(n)$
is contained in $LY(1)$ for any inverse temperature $\beta\ge0 $.
\end{fact}
Since our setup is slightly different from the original work~\cite{sf71}, for the sake of completeness we sketch the proof below\footnote{The original work~\cite{sf71}
considered the quantum partition function $\mathrm{Tr}(e^{-\beta H})$ as a function of a single variable parameterizing
the magnetic $Z$-field and showed that all zeros of this function lie on the imaginary axis  in the complex plane.}.
\begin{proof}
We have to prove that $e^{-\beta H} \in LY(1)$ for any $H\in \mathsf{SF}(n)$ and any $\beta \ge0$.
We shall use the following corollary of Theorem~\ref{thm:sf71}.
\begin{corol}
\label{corol:sf71}
Suppose $F$ is a $k$-qubit operator. Let $F^*$ be the complex conjugate of $F$ in the standard
basis of $k$ qubits. Suppose $F X^{\otimes k} = X^{\otimes k}F^*$ and 
\be
\label{sf_condition}
|\la 0^k|F|0^k\ra| \ge \frac14 \sum_{a,b \in \{0,1\}^k} |\la a|F|b\ra|.
\ee
Then $F\in LY(1)$.
\end{corol}
\begin{proof}
Vectorizing $F$ gives a special case of Theorem~\ref{thm:sf71} with $n=2k$.
\end{proof}
Let 
\[
H_i^z = \mu^z_i Z_i \quad \mbox{and} \quad H_i^{xy} = \mu^x_i X_i + \mu^y_i Y_i.
\]
Write $H=-\sum_{a=1}^m h_a$, where each term $h_a$ is the Hamiltonian $H^z_i$ or $H^{xy}_i$ or  $H_{i,j}$ applied to some qubit or a pair of qubits. 
Pick a large enough integer $\ell$  and let 
$\epsilon = \frac{m\beta}{\ell}$.
The first-order Trotter circuit with $\ell/m$ Trotter steps gives an approximation
\be
\label{TSapprox}
e^{-\beta H} \approx V_\ell \cdots V_2 V_1
\ee
where 
\[
V_t = I + \epsilon h_{f(t)}
\]
for  $f(t)=1+t {\pmod m}$. It is well known that the approximation Eq.~(\ref{TSapprox}) becomes exact in the limit $\ell\to \infty$.
We claim that all individual gates $V_t$ are contained in $LY(1)$ provided that $\epsilon>0$ is small enough.

Indeed, suppose
 $V_t$ is a single-qubit gate $I+\epsilon H_i^z$. Ignoring all  qubits not participating in $V_t$, the generating polynomial of $V_t$ is 
\[
f(u)=\la u_1^*|V_t|u_2\ra = 1+ \epsilon \mu^z  + (1-\epsilon \mu^z) u_1 u_2,
\]
where $\mu^z= \mu^z_i\ge 0$ and $u=(u_1,u_2)\in \CC^2$ are complex variables.
Suppose  $|u_1|<1$ and $|u_2|<1$. Then 
$|f(u)| > 1+\epsilon \mu^z - 1+\epsilon \mu^z = 2\epsilon \mu^z\ge0$
for all small enough $\epsilon>0$. Hence $V_t\in LY(1)$.

Suppose $V_t$ is a single-qubit gate $I+\epsilon H_i^{xy}$.
Ignoring all  qubits not participating in $V_t$ one gets 
\[
V_t = \left[ \ba{cc}
1 & \epsilon( \mu^x - i \mu^y) \\
\epsilon( \mu^x + i\mu^y) & 1\\
\ea\right]
\]
for some real $\mu^x,\mu^y$. One can check that 
 $XV_tX=V_t^*$ and $|\la 0|V_t|0\ra| \ge (1/4)\sum_{a,b\in \{0,1\}} |\la a|V_t|b\ra|$,
 for all small enough $\epsilon>0$.
 Corollary~\ref{corol:sf71} then implies $V_t\in LY(1)$.

Suppose $V_t$ is a two-qubit gate $I+\epsilon H_{i,j}$. 
Ignoring all  qubits not participating in $V_t$ one gets
\[
V_t =  \left[ \ba{cccc}
1 +\epsilon J^z  & 0 & 0 & \epsilon (J^x -J^y) \\
0 & 1 - \epsilon J^z & \epsilon (J^x +J^y) & 0 \\
0 & \epsilon (J^x +J^y)  & 1-\epsilon J^z  & 0 \\
\epsilon (J^x - J^y) & 0 & 0 & 1+\epsilon J^z \\
\ea\right]
\]
for some real $J^x,J^y,J^z$. One can check that $X^{\otimes 2} V_tX^{\otimes 2}=V_t^*$
and   $|\la 00|V_t|00\ra| \ge (1/4)\sum_{a,b\in\{0,1\}^2} |\la a|V_t|b\ra|$ for all $\epsilon>0$
whenever $J^z\ge |J^x|$ and $J^z\ge |J^y|$.
Corollary~\ref{corol:sf71} then implies $V_t\in LY(1)$.
 
 To conclude, the Trotter circuit $V=V_\ell \cdots V_1$ is a (partial) contraction of a tensor network
 composed of $LY(1)$ tensors $V_1,\ldots,V_\ell$ and tensors associated with the identity operators.
 The latter are also $LY(1)$ tensors.
 Lemma~\ref{lemma:contraction} implies that 
either $V=0$ or $V\in LY(1)$. The former case is impossible since $V$ is a product of full-rank operators.
 Thus $V\in LY(1)$. 
  Taking the limit $\ell\to \infty$ and using Corollary~\ref{corol:seq}
 proves $e^{-\beta H} \in LY(1)$ which is equivalent to $\rho_\beta \in LY(1)$.
\end{proof}
{\em Comment:} Suppose $H\in \mathsf{SF}(n)$ and $H'$ is obtained from $H$ by replacing each two-qubit
interaction $H_{i,j}$ defined in Eq.~(\ref{SFmodel2}) with 
\[
H_{i,j}'= U_{i,j} H_{i,j} U_{i,j}^\dag,
\]
where $U_{i,j}$ is a diagonal unitary operator that applies arbitrary $Z$-rotations to the $i$-th and to the $j$-th qubit.
The rotation angles may depend on the graph edge, so that the Hamiltonians $H$ and $H'$ may have different eigenvalues. 
However,
the same proof  as above shows that $e^{-\beta H'} \in LY(1)$ since modifying a tensor by 
applying  single-qubit $Z$-rotations to some indices does not change the Lee-Yang radii. 

\subsection{XXZ Hamiltonians}
\label{subs:XXZ}

Let us now consider a  subclass of Suzuki-Fisher Hamiltonians studied in~\cite{harrow2020classical}.
Define the particle number operator
\[
N=\sum_{i\in V} |1\ra\la 1|_i.
\]
A Hamiltonian $H$ is said to be particle number preserving if $H$ commutes with $N$.
One can easily check that a Hamiltonian 
$H$ defined in Eqs.~(\ref{SFmodel1},\ref{SFmodel2},\ref{SFmodel3}) 
is particle number preserving if 
$J^x_{i,j}=J^y_{i,j}$ and $\mu^x_i=\mu^y_i=0$ for all $i,j$. 
Then
\be
\label{XXZ}
H = -\sum_{(i,j) \in E} J^{xy}_{i,j}(X_i X_j + Y_i Y_j) + J^z_{i,j} Z_i Z_j - \sum_{i\in V} \mu^z_i Z_i
\ee
is a generalization of the XXZ spin chain~\cite{yang1966one}.
From Fact~\ref{fact:sf71} one gets the following.

\begin{corol}
\label{corol:SFradius}
Suppose $H\in \mathsf{SF}(n)$ is particle number preserving.
Then 
\be
e^{-\beta H} \in LY(r)  \quad \mbox{where} \quad r=\min_{i\in V} e^{\beta \mu^z_i}.
\ee
In particular, $r>1$ if all magnetic $Z$-fields are positive and $\beta>0$.
\end{corol}
\begin{proof}
Let $\mu = \min_{i\in V} \mu^z_i$ and $H'$ be a Hamiltonian obtained from $H$ by replacing each term $\mu^z_i Z_i$ with $(\mu^z_i - \mu) Z_i$.
Using the identity $Z = I - 2|1\ra\la 1|$ one gets
$H=cI + H' + 2\mu N$,
where $c=-\mu n$. Since $H$  is particle number preserving, one gets
\be
\label{radius_pushing_by_N}
e^{-\beta H} = \gamma e^{-\beta \mu N} e^{-\beta H'} e^{-\beta \mu N}
\ee
with $\gamma = e^{\beta \mu n}$. Since $H'$ has non-negative magnetic $Z$-fields
and two-qubit terms in $H'$ are the same as those in $H$, 
we have 
$H'\in \mathsf{SF}(n)$.
Fact~\ref{fact:sf71} then gives $e^{-\beta H'} \in LY(1)$.
Let $r=e^{\beta \mu}$.
Using the identity $e^{-\beta \mu |1\ra\la 1|}=|0\ra\la 0| + r^{-1} |1\ra\la 1|$
and Eq.~(\ref{radius_pushing_by_N}) one gets 
$e^{-\beta H} \in LY(r)$. Indeed, if $u,v\in \DD_{r}^n$ then
\[
\la u|e^{-\beta H}|v\ra = \gamma \la u'|e^{-\beta H'}|v'\ra
\]
where $u'=r^{-1} u \in \DD_1^n$ and $v'=r^{-1} v \in \DD_1^n$.
From $\la u'|e^{-\beta H'}|v'\ra\ne 0$ one gets $\la u|e^{-\beta H}|v\ra \ne 0$.
\end{proof}
Unfortunately, particle number preserving Hamiltonians $H\in \mathsf{SF}(n)$
have the trivial ground state $|0^n\ra$.
This can be seen by checking that $|0^n\ra$ is the ground state of each local term $-J^{xy}_{i,j} (X_i X_j + Y_i Y_j) - J^z_{i,j} Z_i Z_j$
provided that $J^z_{i,j}\ge |J^{xy}_{i,j}|$. Furthermore, thermal properties of such Hamiltonians can be simulated efficiently (in time quasi-polynomial in $n$)
using the classical algorithm of Harrow, Mehraban, and Soleimanifar~\cite{harrow2020classical}.
We leave as an open question whether some version of Corollary~\ref{corol:SFradius} holds for general Hamiltonians $H\in \mathsf{SF}(n)$.
For example, is it true that for any $H\in \mathsf{SF}(n)$ with strictly positive magnetic $Z$-fields one has 
$e^{-\beta H} \in LY(r)$ for some $r>1$ ?

\subsection{EPR-like Hamiltonians}
\label{subs:EPR}

The EPR state $(|00\ra + |11\ra)/\sqrt{2}$ is the unique ground state of a
Hamiltonian $-(X_1 X_2 + Z_1 Z_2 - Y_1 Y_2)$. 
Given an interaction graph $G=(V,E)$ as above, 
define an $n$-qubit EPR Hamiltonian as
\be
\label{EPR}
H = -\sum_{(i,j)\in E} X_i X_j + Z_i Z_j - Y_i Y_j.
\ee
It follows directly from the definitions that $H\in \mathsf{SF}(n)$. 
Note also that $H$ is not particle number preserving (it only preserves the number of particles modulo two).
As discussed in Section~\ref{sec:motivation}, the Heisenberg anti-ferromagnetic model
on any bipartite graph is equivalent to the EPR Hamiltonian up to a local basis change on each qubit
that flips the sign of terms $X_iX_j$ and $Z_i Z_j$. 
From Fact~\ref{fact:sf71} one gets $e^{-\beta H} \in LY(1)$ for all $\beta \ge 0$.
However,  $e^{-\beta H} \notin LY(r)$ for any $r>1$.  Indeed,
otherwise one would have $\la u|e^{-\beta H}|v\ra\ne 0$ for any
vectors $u,v\in \CC^n$ such that $|u_p|=|v_p|=1$ for all $p$.
Choose  $|u\ra$ and $|v\ra$ as tensor products of single-qubit states
$|0\ra \pm |1\ra$ such that  $X^{\otimes n}|u\ra=|u\ra$ 
and $X^{\otimes n}|v\ra=-|v\ra$. Then 
\[
\la u|e^{-\beta H} |v\ra = \la u|X^{\otimes n} e^{-\beta H} |v\ra = \la u|e^{-\beta H} X^{\otimes n} |v\ra = -\la u|e^{-\beta H} |v\ra.
\]
Here we used the fact that $H$ commutes with $X^{\otimes n}$.
Hence $\la u|e^{-\beta H}|v\ra= 0$ which is a contradiction. 

Fix a parameter $s\in [0,1]$ and define an EPR-like Hamiltonian
\be
H_s = -\sum_{(i,j)\in E} |\phi_s\ra\la \phi_s|_{i,j}
\ee
where
\be
|\phi_s\ra = |00\ra + s|11\ra
\ee
and $|\phi_s\ra\la \phi_s|_{i,j}$ denotes the operator $|\phi_s\ra\la \phi_s|$ applied to qubits $i,j$ tensored with the identity on all other qubits.
One can easily check that 
\be
|\phi_s\ra\la \phi_s|_{i,j} = \frac14\left(
(1+s^2)I + (1+s^2) Z_i Z_j + (1-s^2)(Z_i + Z_j) + 2s X_i X_j - 2s Y_i Y_j\right).
\ee
Thus $H_1$ coincides with the standard EPR Hamiltonian Eq.~(\ref{EPR}) up to rescaling and an energy shift. 
Since we assumed $s\in [0,1]$, one has $1+s^2\ge 2s$.
Thus $H_s$ has ferromagnetic dominating $Z_i Z_j$ terms and non-negative magnetic $Z$-fields,
that is,
\be
H_s \in \mathsf{SF}(n) \quad \mbox{for all $s\in [0,1]$},
\ee
ignoring an overall energy shift. In particular, $e^{-\beta H_s} \in LY(1)$.  

From now on we assume that $0\le s<1$ and $G$ is a connected graph. 
Is it true that $e^{-\beta H_s} \in LY(r)$ for
some $r>1$ ? Consider first the simplest case when $n=2$. 
\begin{lemma}
\label{lemma:radius_n2}
Suppose $n=2$ and the graph $G$ is a single edge. Then 
$e^{-\beta H_s} \in LY(r)$ where
\be
\label{radius_n2}
r=\frac1{\sqrt{s}} \left(
\frac{1+s^2 e^{-\beta (1+s^2)}}{1+ s^{-2} e^{-\beta(1+s^2)}}\right)^{1/4}.
\ee
Note that $r>1$ for any $s\in [0,1)$ and $\beta>0$.
Furthermore, $e^{-\beta H_s} \notin LY(r')$ for any $r'>r$.
\end{lemma}
\begin{proof}
We have
\[
e^{-\beta H_s} = I + \epsilon |\phi_s\ra\la \phi_s|,\qquad \epsilon = \frac{e^{\beta(1+s^2)}-1}{1+s^2}.
\]
Note that $\epsilon\ge0$.
The generating polynomial of $e^{-\beta H_s}$ is 
\[
f(z)=\la z_1^*,z_2^*|e^{-\beta H_s}|z_3,z_4\ra =  (1+z_1z_3)(1+z_2z_4) + \epsilon (1+ s z_1 z_2)(1+ s z_3 z_4).
\]
Write $f(z)=\det{M(z)}$, where 
\be
M(z)=\left[ \ba{cc}
1+z_1 z_3 & -\epsilon^{1/2} (1+sz_1 z_2) \\
\epsilon^{1/2} (1+sz_3z_4) & 1+z_2 z_4\\
\ea
\right].
\ee
We have 
\[
M(z) = \left[ \ba{cc}
1 & -\epsilon^{1/2} \\
\epsilon^{1/2} & 1 \\
\ea
\right] 
+ \left[
\ba{cc}
z_1 & 0 \\
0 & z_4 \\
\ea
\right]
\cdot
\left[ \ba{cc}
1& -s\epsilon^{1/2} \\
s\epsilon^{1/2} & 1 \\
\ea
\right] 
\cdot
 \left[
\ba{cc}
z_3 & 0 \\
0 & z_2 \\
\ea
\right].
\]
Clearly,
\[
 \left[ \ba{cc}
1 & -\epsilon^{1/2} \\
\epsilon^{1/2} & 1 \\
\ea
\right]  = I - i \epsilon^{1/2} Y =  (1+\epsilon)^{1/2} e^{-i\theta Y}
\]
where $\cos{\theta}=(1+\epsilon)^{-1/2} $. We arrive at 
\[
\det{M(z)} = (1+\epsilon) \det{\left(I+K(z)\right)},
\]
where 
\[
K(z) =  (1+\epsilon)^{-1/2} e^{i\theta Y}
 \left[
\ba{cc}
z_1 & 0 \\
0 & z_4 \\
\ea
\right]\cdot
\left[ \ba{cc}
1& -s\epsilon^{1/2} \\
s\epsilon^{1/2} & 1 \\
\ea
\right] 
\cdot
 \left[
\ba{cc}
z_3 & 0 \\
0 & z_2 \\
\ea
\right].
\]
Let $r>0$ be a radius to be defined later. 
For any $z_1,z_2,z_3,z_4\in \DD_r$  one has
\[
\left\| 
\left[
\ba{cc}
z_1 & 0 \\
0 & z_4 \\
\ea
\right]
\right\|
< r, \quad
\left\| \left[
\ba{cc}
z_3 & 0 \\
0 & z_2 \\
\ea
\right]
\right\| < r,
\quad
\left\|
\left[  \ba{cc}
1 & -s\epsilon^{1/2} \\
s\epsilon^{1/2} & 1\\
\ea\right]
\right\| = \sqrt{1+s^2 \epsilon}.
\]
Sub-multiplicativity of the norm implies that 
\[
\| K(z)\| < r^2  \left( \frac{1+s^2 \epsilon}{1+\epsilon}
\right)^{1/2}\le 1
\]
if we choose
\be
\label{r_from_epsilon}
r=\left( \frac{1+\epsilon}{1+s^2 \epsilon}\right)^{1/4}.
\ee
From  $\|K(z)\|<1$ one gets
$ \det{M(z)}\ne 0$ and thus
$f(z)\ne 0$ for all $z_1,z_2,z_3,z_4\in \odisk{r}$.
Substituting the definition of $\epsilon$ into Eq.~(\ref{r_from_epsilon}) 
gives Eq.~(\ref{radius_n2}).

The last statement of the lemma follows from 
\[
f(r,-r,-r,-r)=(1-r^2)(1+r^2) + \epsilon (1-sr^2)(1+sr^2) = 1+\epsilon - r^4(1+\epsilon s^2)=0.
\]
\end{proof}
{\em Comment:} Numerical experiments 
for small ($n\le 6$) random examples suggest that 
the first part of
Lemma~\ref{lemma:radius_n2} might hold for all connected graphs, that is, $e^{-\beta H_s} \in LY(r)$, where $r$ is determined by Eq.~(\ref{radius_n2}) for all $\beta\ge 0$, $s\in [0,1]$, and all connected graphs $G$.

Let $|\psi\ra$ be the ground state of $H_s$.
In the next section we show $H_s$
has a non-degenerate ground state whenever $n\ge 2$ and the interaction graph $G$ is connected.  Taking the limit $\beta\to \infty$ in Lemma~\ref{lemma:radius_n2} gives $\psi \in LY(r)$ with $r=1/\sqrt{s}$ for $n=2$.
Numerical experiments suggest the following. 
\begin{conj}
\label{conj:lygap} 
Suppose $n\ge 2$,
$s\in [0,1)$,
and the interaction graph is connected. 
Let $\ket{\psi}$ be the ground state $H_s$. Then there exists $r>1$ such that $\ket{\psi}\in LY(r).$ 
More ambitiously, $\ket{\psi}\in LY(s^{-1/2})$ and $\ket{\psi}\notin LY(r)$ for any $r>s^{-1/2}.$ 
\end{conj}
In particular, our numerical evidence for this conjecture is as follows.
\begin{numerics}[\textbf{Lee-Yang radius of EPR-like Hamiltonians}] 
\label{num:ly_gap}For each connected graph with $3\leq n\leq 8$ vertices and trees with $3\leq n\leq 13$ vertices, we sampled $10$ random values of $s\in (0,1)$ and computed the roots of the univariate polynomials associated with equatorial postselection
\begin{equation}
    f_y(z)=\bigotimes_{i=1}^n(\bra{0}+(-1)^{y_i}z\bra{1})\ket{\psi}
\end{equation} 
for each $y\in\{0,1\}^n.$
 We found that none of the sampled polynomials have roots $z$ with $|z|<s^{-1/2}.$
\end{numerics}

\begin{figure}[H]
    \centering
    \makebox[\paperwidth][l]{\includegraphics[width=0.49\paperwidth]{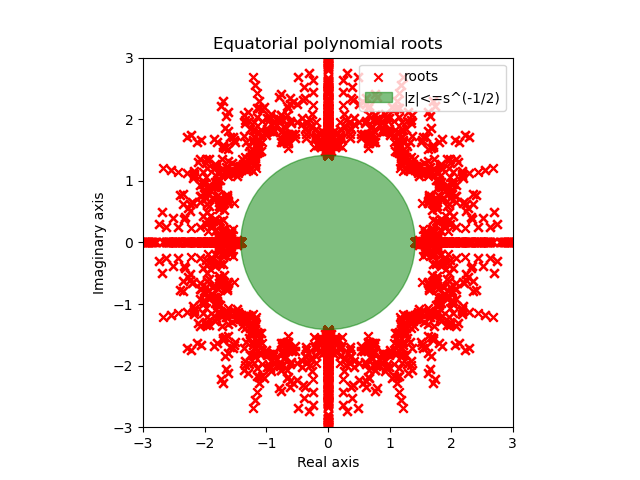}
    \hspace{-2cm}\includegraphics[width=0.49\paperwidth]{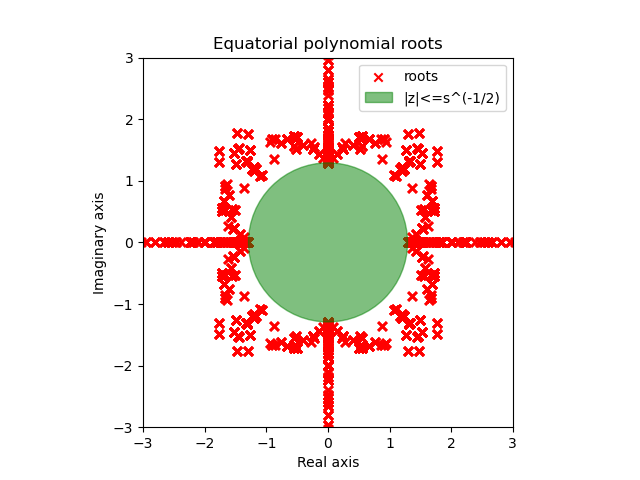}}
\caption{Plots from Numerical Study~\ref{num:ly_gap}. Zeros of equatorial polynomials which lie near the disk $\DD_{s^{-1/2}}$ for a path and cycle  with $10$ vertices and $s=0.5,0.6$ respectively. }
    \label{fig:ly_numerics}
\end{figure}

\subsection{Spectral gap of EPR-like Hamiltonians}
\label{subs:EPRgap}

Theorem~\ref{thm:degeneracy} provides a link between the eigenvalue spectrum of a Hamiltonian $H$ and the
Lee-Yang radius $r$ of the Gibbs state $\rho\sim e^{-\beta H}$ with $\beta>0$. Namely, $r>1$ implies that the maximum eigenvalue of $\rho$
is non-degenerate.
Equivalently, if $\lambda_1\le \lambda_2\le \ldots \le \lambda_{2^n}$ are the eigenvalues of $H$ then
the spectral gap $\lambda_2-\lambda_1$ is strictly positive. 
  In this section we further explore the interplay  between the Lee-Yang radius and the spectral gap
for EPR-like Hamiltonian $H_s$
defined in Section~\ref{subs:EPR}.

\begin{lemma}
\label{lemma:lambda1}
Suppose $n\ge 2$, $s\in [0,1)$,
and the interaction graph is connected. Then 
the ground state of  $H_s$  is non-degenerate.
The ground state can be chosen as  a superposition of even-weight bit strings
with non-negative amplitudes. 
\end{lemma}
\begin{proof}
Let $h$ be the Hadamard gate. 
Since $s\in [0,1)$,   the state
\[
h\otimes h |\phi_s\ra=(1/2)(1+s)(|00\ra+|11\ra) + (1/2)(1-s)(|01\ra+|10\ra)
\]
has  positive amplitudes in the standard basis.
It follows that $h^{\otimes 2}|\phi_s\ra\la \phi_s|h^{\otimes 2}$ is a matrix
with positive entries.
Since the interaction graph  is connected and $n\ge 2$, each vertex has at least one incident edge.
Thus, $-H_s$ is a non-negative irreducible matrix in the Hadamard basis.   Perron-Frobenius theorem implies that the smallest
eigenvalue of $H_s$ is non-degenerate and the corresponding eigenvector $|\psi\ra$  has real positive  amplitudes in the Hadamard basis.
Since the state $|0^n\ra$ also has real positive amplitudes in the Hadamard basis, one infers that $\la 0^n|\psi\ra>0$. 
From $H_sZ^{\otimes n}=Z^{\otimes n}H_s$ and the fact that the smallest  eigenvalue is non-degenerate one infers that $Z^{\otimes n}|\psi\ra = \pm |\psi\ra$.
Combining this and $\la 0^n|\psi\ra>0$ gives $Z^{\otimes n}|\psi\ra =|\psi\ra$, that is, $|\psi\ra$ is a superposition of even-weight bit strings
with non-negative amplitudes.
\end{proof}

Let $\lambda_j$ be the $j$-th smallest eigenvalue of $H_s$, where $j=1,\ldots,2^n$.
We are interested in the spectral gap $\lambda_2-\lambda_1$ which controls the cost of preparing the ground state of $H_s$
by quantum adiabatic evolution starting from $s=0$ (note that $H_0$ has the unique ground state $|0^n\ra$). 
Since $H_s$ commutes with $Z^{\otimes n}$, eigenvectors of $H_s$ can be chosen such that any eigenvector $\ket{\phi}$ is supported exclusively on the even-weight or the odd-weight subspaces corresponding to $\bra{\phi}Z^{\otimes n}\ket{\phi}=1$ and $\bra{\phi}Z^{\otimes n}\ket{\phi}=-1$ respectively. Let $\lambda_j^e$ and $\lambda_j^o$ be the $j$-th smallest eigenvalue of $H_s$ within the even-weight
and the odd-weight subspace respectively. Here $j=1,\ldots,2^{n-1}$. 
From Lemma~\ref{lemma:lambda1} one gets $\lambda_1=\lambda_1^e<\lambda_1^o$.
Thus the spectral gap of $H_s$ is 
\be
\label{gap_even_odd}
\lambda_2-\lambda_1 = \min{\{\lambda_1^o-\lambda_1^e, \lambda_2^e-\lambda_1^e\}}.
\ee
Consider first the simple case when the interaction graph $G$ is a star graph such that 
\[
H_s = -\sum_{i=2}^n |\phi_s\ra\la \phi_s|_{1,i}.
\]
\begin{lemma}
\label{lemma:star}
Suppose $G$ is the star graph with $n\ge 3$ vertices. Then the EPR-like Hamiltonian $H_s$
has spectral gap
$\lambda_2-\lambda_1=\lambda_1^o-\lambda_1^e=1-s^2$.
\end{lemma}
\begin{proof}
Let us first compute $\lambda_1^e$ and $\lambda_1^o$.
Let $H^e$ and $H^o$ be the restrictions of $-H_s$ onto the even-weight and the odd-weight subspaces.
Since $H^e$ and $H^o$ are entry-wise non-negative Hamiltonians,
Perron Frobenius theorem implies that maximum eigenvectors of $H^e$ and $H^o$ can be chosen as real vectors
with non-negative amplitudes in the standard basis. Furthermore, since $H^e$ and $H^0$ are invariant under any permutation of qubits $2,\ldots,n$,
one can symmetrize the maximum eigenvectors of $H^e$ and $H^o$ over such permutations (the symmetrization does not annihilate the state
since all amplitudes are non-negative). 
Thus it suffices to compute maximum eigenvalues of  $H^e$ and $H^o$ restricted to a subspace 
\be
\calH = \CC^2 \otimes \calS_{n-1},
\ee
where $\calS_{n-1} \subseteq (\CC^2)^{\otimes (n-1)}$ is the symmetric subspace of $n-1$ qubits. 
By definition, $\calS_{n-1}$ is spanned by Dicke states 
\be
|D_k\ra = \frac1{\sqrt{{n-1 \choose k}}} \sum_{\substack{x\in \{0,1\}^{n-1}\\ |x|=k}}\; |x\ra.
\ee
After some tedious algebra one gets
\be 
-H_s |0\ra \otimes |D_k\ra =  (n-1-k) |0\ra \otimes |D_k\ra + \sqrt{s^2 (k+1)(n-k-1)} |1\ra\otimes |D_{k+1}\ra
\ee
and
\be 
-H_s |1\ra \otimes |D_{k+1}\ra  =s^2 (k+1)  |1\ra \otimes |D_{k+1}\ra +  \sqrt{s^2 (k+1)(n-k-1)} |0\ra\otimes |D_{k}\ra.
\ee
Hence $-H_s$ preserves the two-dimensional subspace spanned by $|0\ra \otimes |D_k\ra$ and $|1\ra \otimes |D_{k+1}\ra$.
The corresponding $2\times2$ block of $-H_s$ is a rank-1 matrix
\be
\left[ \ba{cc}
\alpha^2 & \alpha \beta \\
\alpha \beta & \beta^2 \\
\ea\right]
\ee
where $\alpha=\sqrt{n-1-k}$ and $\beta = \sqrt{s^2  (k+1)}$. The maximum eigenvalue of this $2\times 2$ matrix coincides with its trace, that is,
$\alpha^2+\beta^2 = n-1-k + s^2(k+1)$. Maximizing this expression 
over even and odd values of $k\ge 0$ and noting that $s\in [0,1)$ gives
\[
-\lambda_1^e =n-1+s^2, \quad -\lambda_2^e = n-3+3s^2, \quad -\lambda_1^o=n-2+2s^2.
\]
Hence 
\[
\lambda_1^o-\lambda_1^e=1-s^2 \quad \mbox{and} \quad \lambda_2^e - \lambda_1^e = 2(1-s^2).
\]
The statement of the lemma now follows from  Eq.~(\ref{gap_even_odd}).
\end{proof}
Numerical experiments suggest the following two conjectures.
\begin{conj}
\label{conj:lambda2}
Suppose $n\ge 2$, $s\in [0,1)$,
and the interaction graph is connected.
Then the  second smallest eigenvalue of $H_s$ coincides with the smallest  eigenvalue of $H_s$ restricted to the odd-weight subspace.
\end{conj}
\begin{conj}
\label{conj:sgap}
Suppose $n\ge 2$, $s\in [0,1)$,
and the interaction graph is connected.
Then  the spectral gap between the smallest and the second smallest eigenvalues of $H_s$
is at least $1-s^2$.
This lower bound is attained on star graphs with $n\ge 3$ vertices. 
\end{conj}

The following summarizes our numerical evidence for Conjecture \ref{conj:sgap}.
\begin{numerics}[\textbf{Spectral gap of EPR-like Hamiltonians}]\label{num:sgap}
    We numerically investigated Conjecture~\ref{conj:sgap} and found by exact diagonalization that it holds for $100$ uniformly spaced values of $s\in(0,1)$ for all trees with $n\leq 18$ vertices and all connected graphs with $n\leq 10$ vertices.
     We also found that it holds for
    path graphs with $n\le 100$ vertices and $s\in \{0.9, 0.95, 0.97, 0.99\}$.
     The spectral gap for path graphs was computed by the exact diagonalization if $n\le 24$ and using the Matrix Product State library $\mathsf{mpnum}$~\cite{mpnum} for larger values of $n$.
\end{numerics}

\begin{figure}[H]
    \centering
    \makebox[1\paperwidth][l]
{
\hspace{-1.9cm}\includegraphics[width=0.53\paperwidth]{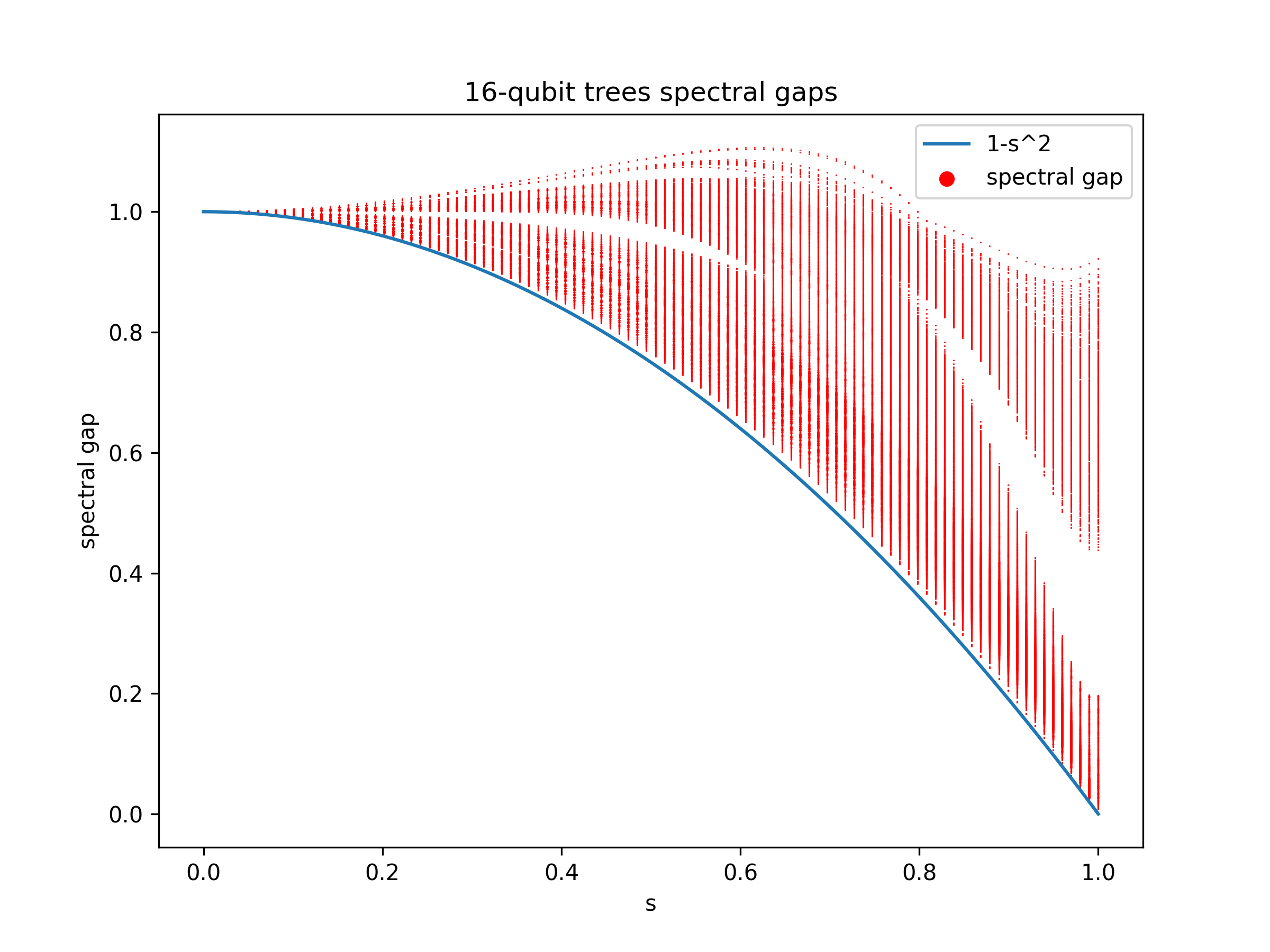}
\hspace{-1cm}\includegraphics[width=0.53\paperwidth]{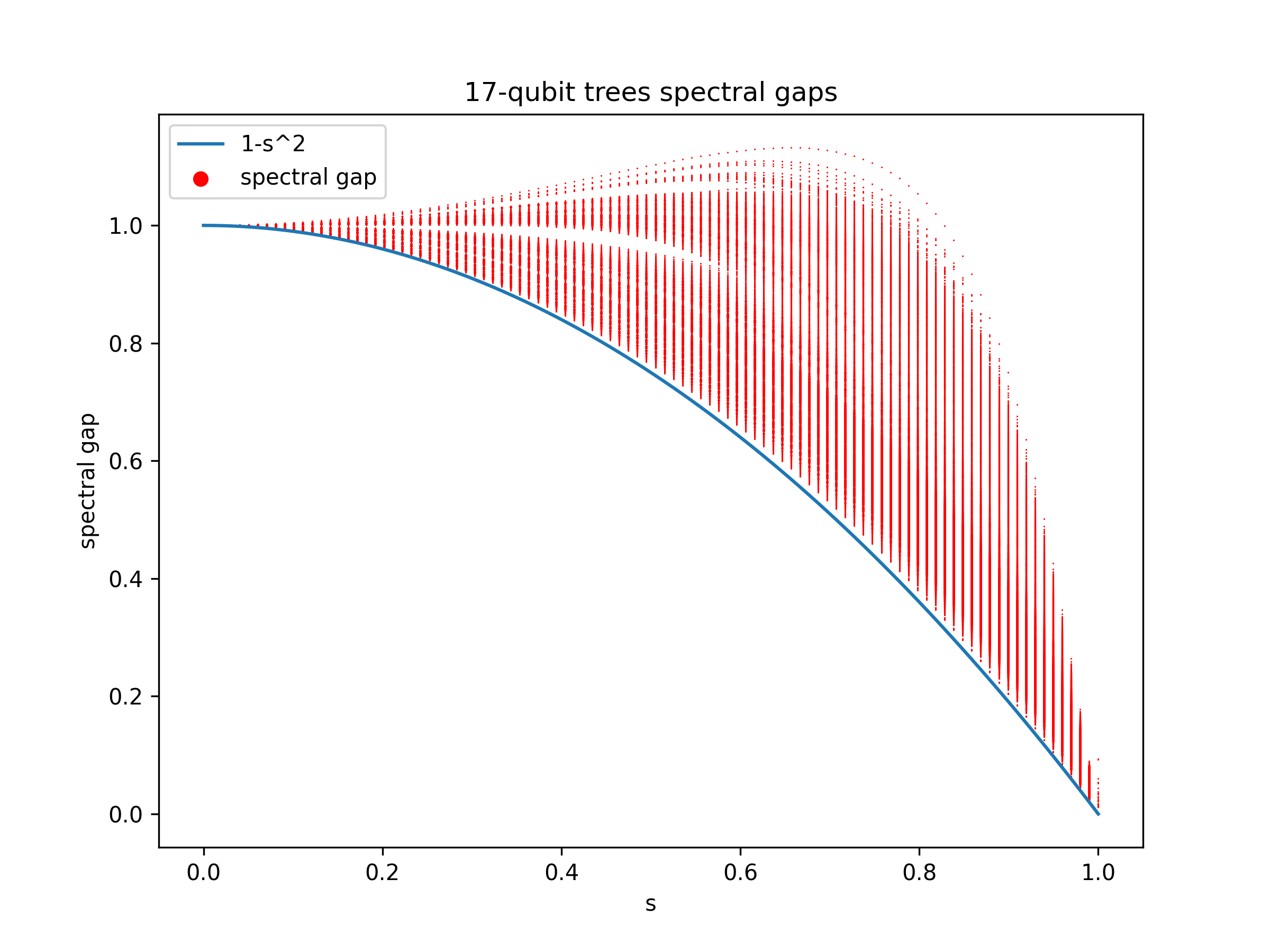}}
\caption{Plots from Numerical Study~\ref{num:sgap}. For each tree on $n=16,17$ vertices, the spectral gap is plotted as a red point for $100$ uniformly spaced $s\in(0,1)$.}
    \label{fig:tree_numerics}
\end{figure}

\begin{figure}[H]
    \centering
    \includegraphics[width=0.49\linewidth]{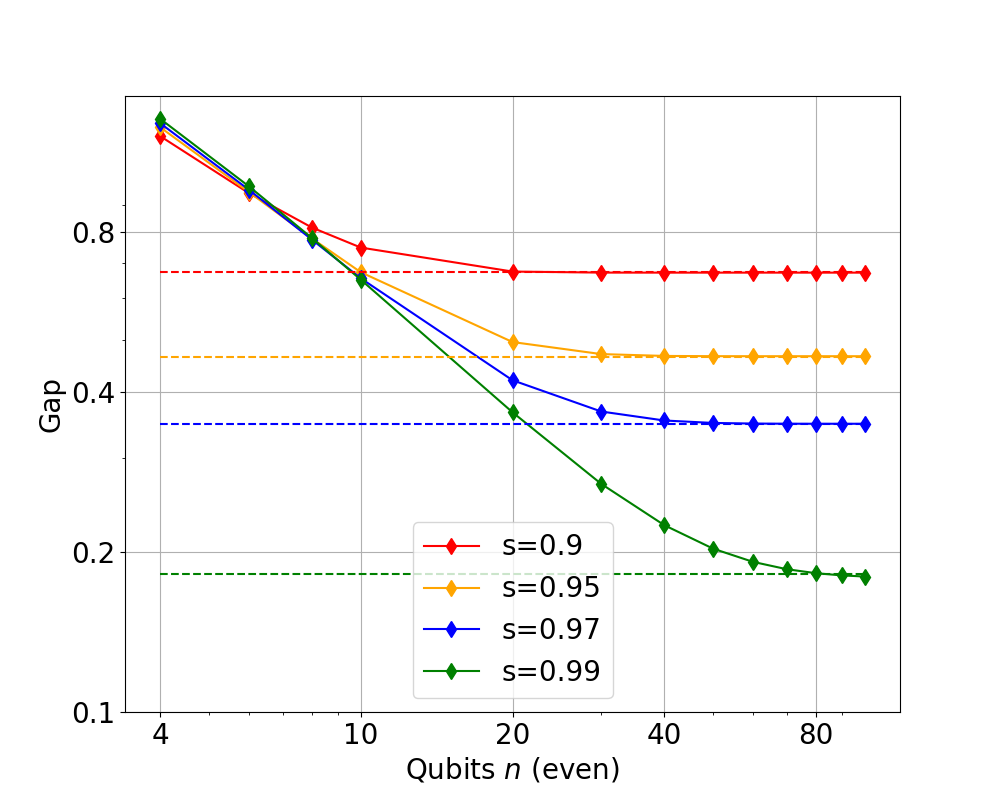}
    \includegraphics[width=0.49\linewidth]{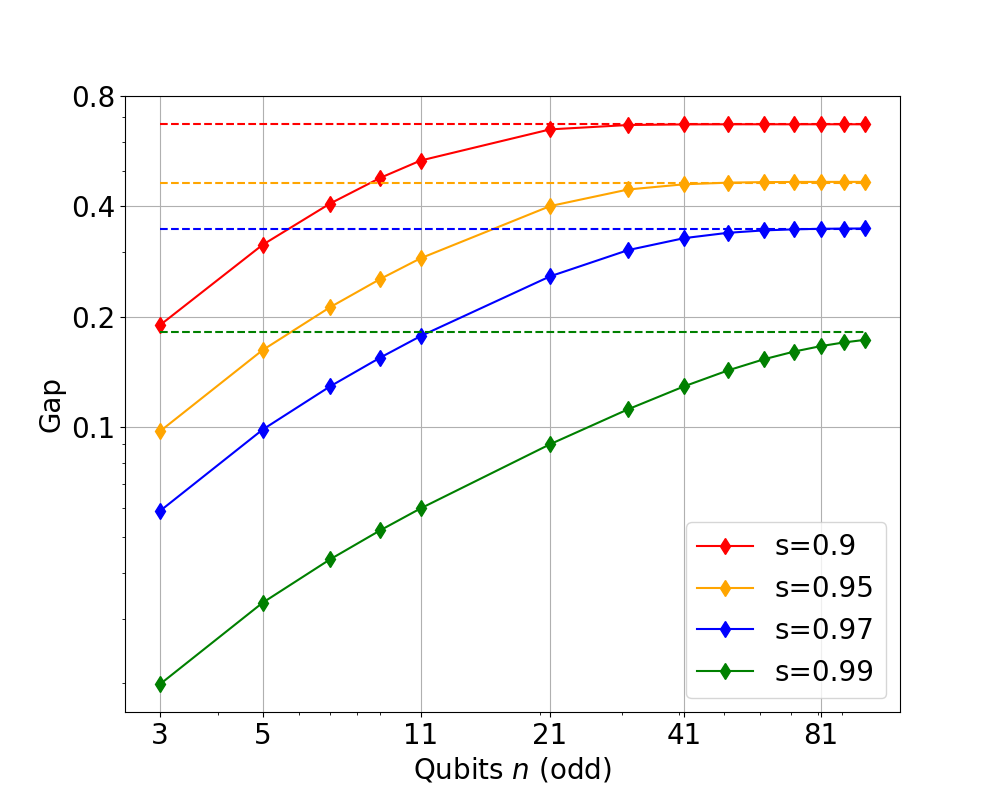}
\caption{Plots from Numerical Study~\ref{num:sgap}.
Spectral gap of  1D chains with $3\le n\le 100$ qubits and open boundary conditions.  The lower bound $1-s^2$ is achieved at $n=3$ (when the chain is the star graph). Horizontal dashed lines on both plots show a fit $g(s) = 3.03 s (1-s)^{0.609}$ which corresponds to the large $n$ limit
(parity of $n$ does not matter in this limit).}
    \label{fig:epr_gap_mps}
\end{figure}

We note that Conjecture~\ref{conj:sgap} is easy to falsify
since it gives a precise bound on the spectral gap for any finite $n$, as opposed to an asymptotic scaling like $poly(1/n)$. 
At the same time, the conjecture is extremely bold. If true,
it would imply that the Bipartite Quantum Max-Cut  problem is contained in BQP
resolving a long standing open question. Indeed, the Quantum Max-Cut problem boils down
to approximating the ground energy of the Heisenberg anti-ferromagnetic Hamiltonian
$H_{AF} =\sum_{(i,j)\in E} X_i X_j + Y_i Y_j + Z_i Z_j$ on a given graph $G=(V,E)$
within an additive error $\epsilon=poly(1/n)$. 
As commented above, $H_{AF}$ coincides with the EPR Hamiltonian $H_1$ (up to a local basis change on some qubits)
whenever the interaction graph is bipartite. Ground energies of $H_1$ and $H_s$ are related by Weyl's inequality
\[
|\lambda_1(H_1)-\lambda_1(H_s)|\le \| H_1 - H_s\| \le \sum_{(i,j)\in E} \|\,  |\phi_s\ra\la \phi_s| - |\phi_1\ra\la \phi_1|\, \| \le 2|E|(1-s).
\]
Suppose $s=1-\epsilon/(4|E|)$ so that $|\lambda_1(H_1)-\lambda_1(H_s)|\le \epsilon/2$.
Let $\tilde{\lambda}$ be an $(\epsilon/2)$-approximation of 
the ground energy of $H_s$ found by the quantum adiabatic evolution algorithm.
Then $\tilde{\lambda}$ is the desired $\epsilon$-approximation of $\lambda_1(H_1)=\lambda_1(H_{AF})$.
Choose the adiabatic path as $\{ H_t\}_{0\le t\le s}$. Recall that $H_0$ has the trivial ground state 
$|0^n\ra$. Then the  quantum runtime scales
as $T=\epsilon^{-1} \cdot poly(g_{min}^{-1},n)$, where $g_{min}$ is the minimum spectral gap along the chosen adiabatic path, see~\cite{benseny2021adiabatic} and references therein. 
Conjecture~\ref{conj:sgap} gives $g_{min}=1-s^2 \sim \epsilon/|E|$. Hence  $T\le poly(n,\epsilon^{-1})$
scales polynomially with all parameters.

\subsection{Phase-shifted EPR-like Hamiltonians}

One limitation of EPR-like Hamiltonians $H_s$ defined in Section~\ref{subs:EPR} is that these Hamiltonians
are stoquastic (sign-problem-free). One might expect that the ground energy of such Hamiltonians can be estimated for 
a moderately large number of qubits (say, $n\le 100$)
on a classical computer using  Quantum Monte Carlo algorithms. In this section we briefly discuss a non-stoquastic version
of EPR-like Hamiltonians where local phase shifts are applied on every edge of the graph.
Such Hamiltonians may be more suitable for demonstrating a quantum advantage for the ground energy estimation problem. 
Define a two-qubit operator
\[
h(z) = |00\ra\la 00| + z |11\ra\la 00| + z^* |00\ra\la 11| + |z|^2 |11\ra\la 11|,
\]
where $z\in \CC$.
A phase-shifted EPR-like   Hamiltonian is defined as
\be
H_s(\theta) = -\sum_{(p,q)\in E} h_{p,q}(se^{i\theta_{p,q}}),
\ee
where $s\in [0,1]$ and 
$\theta_{p,q}\in [0,2\pi)$ are arbitrary edge-dependent angles (phase shifts). Setting $\theta_{p,q}=0$ 
gives the original EPR-like Hamiltonian $H_s$. If the interaction graph $G$ is a tree, phase shifts do not matter in the sense that 
$H_s(\theta)= U H_s U^\dag$, where $U$ is a tensor product of $n$ single-qubit $Z$-rotations\footnote{This can be proved using
the identity $h(se^{i\varphi})=(R_z\otimes I)h(s) (R_z^\dag \otimes I)$ with $R_z=e^{-i(\varphi/2)Z}$ and removing phase shifts 
by a sequence of $R_z$ rotations starting from root of the tree towards the leaves.}.
In particular, the eigenvalue spectrum 
of $H_s(\theta)$
and the Lee-Yang radius of the Gibbs state $\rho \sim e^{-\beta H_s(\theta)}$
do not depend on $\theta$. However, if the interaction graph contains cycles,
phase shifts cannot be generally removed by a global similarity transformation (this is analogous to magnetic fluxes that cannot be removed by a gauge transformation). 
In this case the eigenvalue spectrum and the Lee-Yang radius
of the Gibbs state  may depend on $\theta$.

Since $H_s(\theta)$ is  obtained from
a Hamiltonian $H_s\in \mathsf{SF}(n)$ by conjugating each local term in $H_s$ with single-qubit $Z$-rotations,
one concludes that 
\be
\label{phase_shifted_LY1}
e^{-\beta H_s(\theta)} \in LY(1)
\ee
for any angles  $\theta$, see  the discussion in Section~\ref{subs:SF}.
Taking the limit $\beta \to \infty$ one concludes that the ground state  projector for $H_s(\theta)$
is a tensor in $LY(1)$ and thus at least one ground state of $H_s(\theta)$ has a non-zero overlap with $|0^n\ra$.
As a consequence, one can choose a ground state of $H_s(\theta)$ as 
\[
|\psi\ra = \lim_{\beta \to \infty} \frac{e^{-\beta H_s(\theta)} |0^n\ra}{\| e^{-\beta H_s(\theta)} |0^n\ra\|}.
\]
Furthermore,  the proof of Lemma~\ref{lemma:radius_n2} works for any phase shift  implying that 
$e^{-\beta H_s(\theta)} \in LY(r)$ for $n=2$, where the radius $r$ is given by Eq.~(\ref{radius_n2}).
We leave as an open question whether Conjectures~\ref{conj:lygap},\ref{conj:lambda2},\ref{conj:sgap} hold
for phase-shifted EPR-like  Hamiltonians.  If this is the case, Theorem~\ref{thm:degeneracy} implies that the ground state of $H_s(\theta)$ is non-degenerate
and the adiabatic path connecting $H_0$ with $H_s(\theta)$ has the minimum spectral gap $1-s^2$.
Accordingly, the problem of estimating the ground energy of $H_1(\theta)$ is contained in BQP, see the discussion at the end of Section~\ref{subs:EPRgap}

\begin{numerics}[\textbf{Spectral gap of phase-shifted EPR-like Hamiltonians}]\label{num:sgap_theta}
For each connected graph with $n\leq 10$ vertices we selected phase shifts $\theta_{p,q}$ uniformly and independently at random to construct a Hamiltonian $H_s(\theta)$. We computed the spectral gap of $H_s(\theta)$ using exact diagonalization and found that it is lower bounded
by $1-s^2$,  see Fig.~\ref{fig:phase_shifted}.
\end{numerics}

\begin{figure}[h]
    \centering
    \makebox[\paperwidth][l]{\hspace{-1.9cm}\includegraphics[width=0.53\paperwidth]{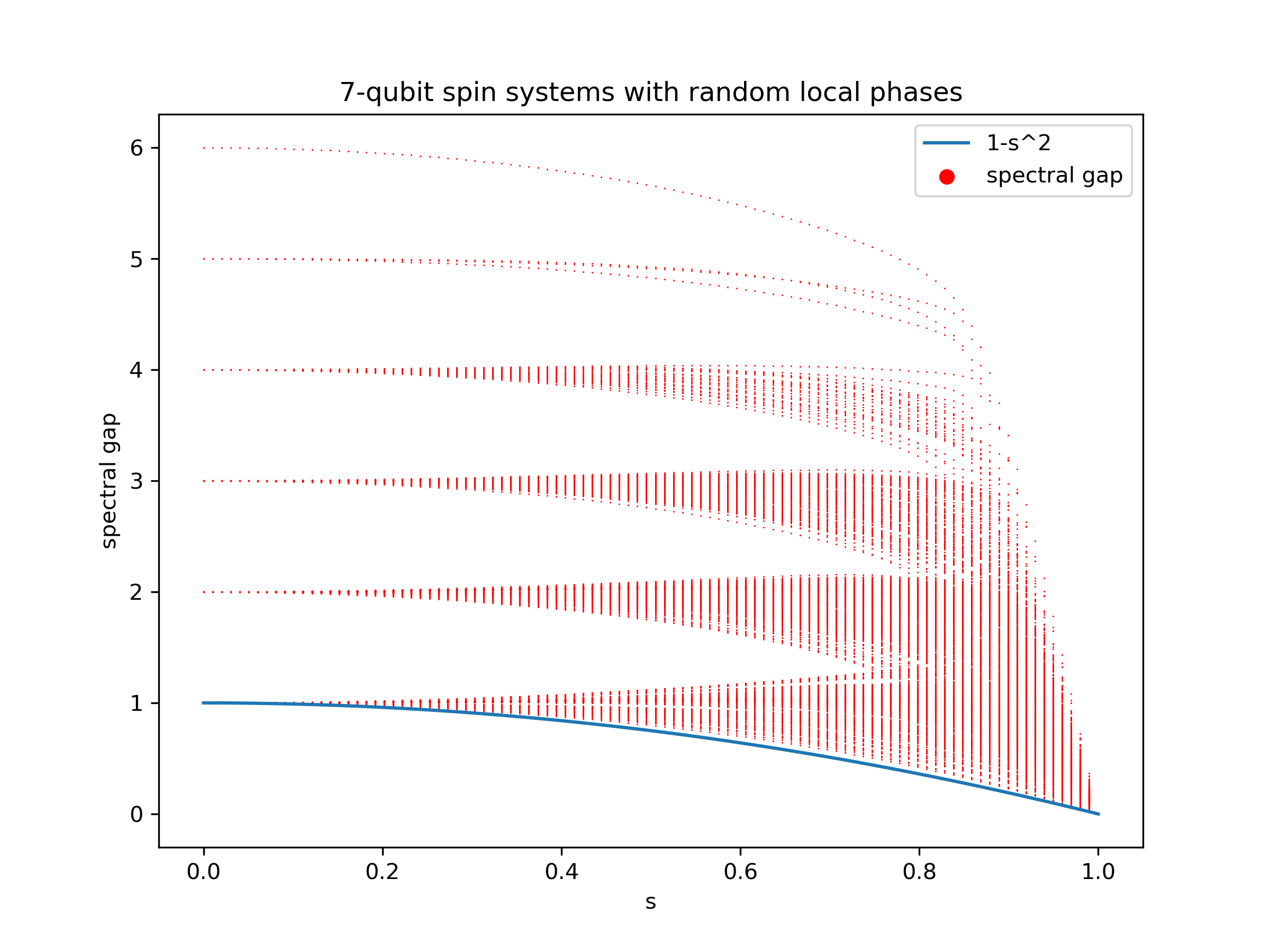}
    \hspace{-1cm}\includegraphics[width=0.53\paperwidth]{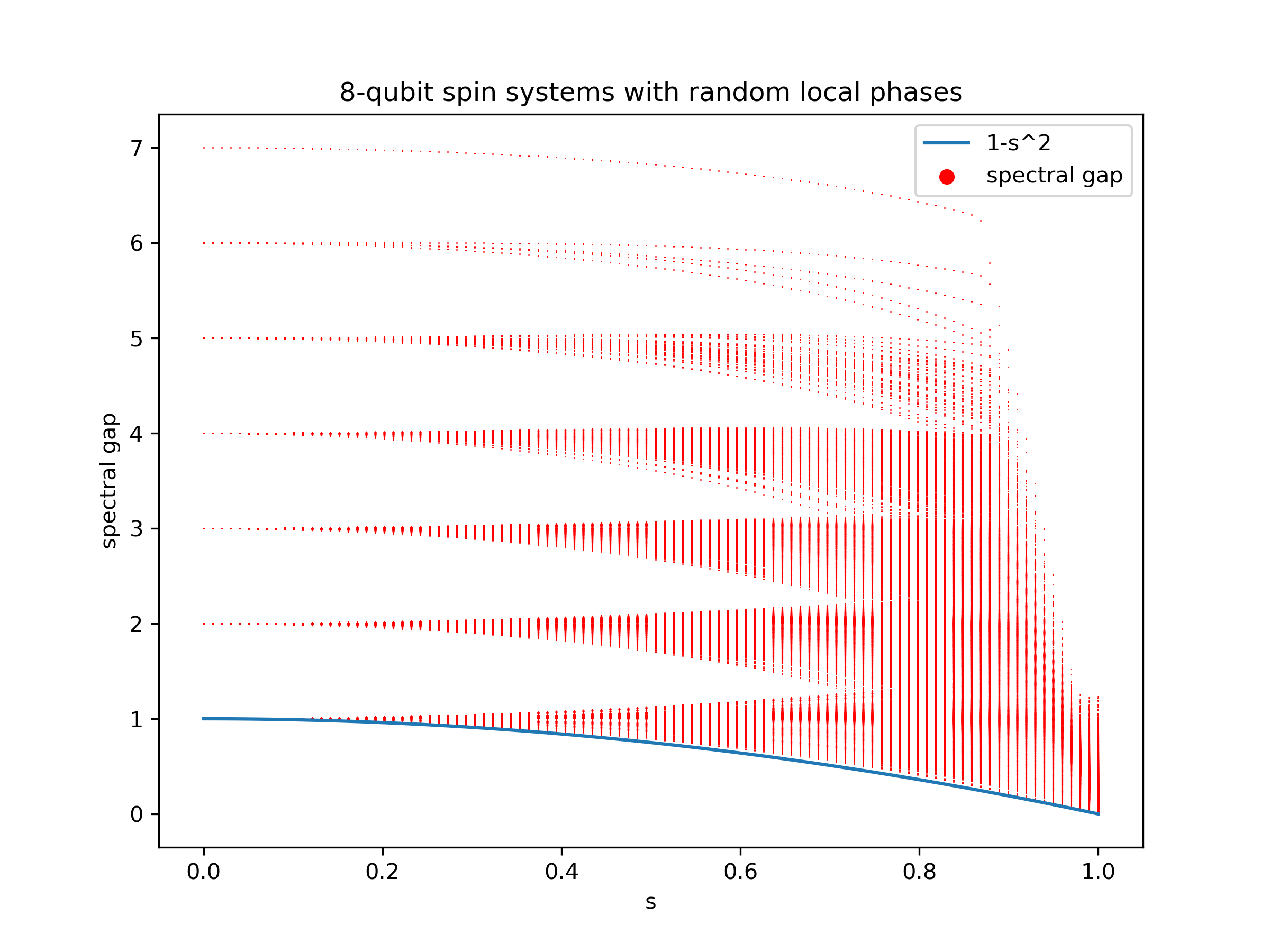}}
\caption{Plots from Numerical Study~\ref{num:sgap_theta}. For each connected graph on $7,8$ vertices, random phase shifts $\theta_{p,q}$ are sampled and the spectral gap of $H_s(\theta)$ is plotted as a red dot for $100$ uniformly spaced $s\in(0,1)$.}
    \label{fig:phase_shifted}
\end{figure}

\section{Open questions}
\label{sec:open}

\noindent\textbf{Beyond Suzuki-Fisher Hamiltonians} All examples of many-body systems with the Lee-Yang property considered above fall in the class of
Suzuki-Fisher Hamiltonians, possibly augmented with local phase shifts. Are there other examples of local Hamiltonians $H$
such that the Gibbs state $\rho \sim e^{-\beta H}$ has Lee-Yang radius $r\ge 1$ 
for any inverse temperature  $\beta \ge 0$ ? Of particular interest are $k$-local Hamiltonians with $k\ge 3$.
Proving the containment $e^{-\beta H}\in LY(1)$ using Trotter approximation of the Gibbs state for $k\ge 3$
appears to require a stronger version of Theorem~\ref{thm:sf71}. 
For example, consider a single Trotter gate $g=I+\epsilon h$ for some
$k$-qubit interaction $h$ and a Trotter step $\epsilon>0$.
The last condition of  Theorem~\ref{thm:sf71} demands
\be
\label{SFcondition_restated}
|\langle 0^{k}|g|0^k\rangle|\ge \frac{1}{4}\sum_{x,y\in \{0,1\}^{k}} |\langle x|g|y\rangle|.
\ee
However, this is impossible  if $k\ge 3$. Indeed,
the righthand side of Eq.~(\ref{SFcondition_restated}) is at least  $2^{k-2} -O(\epsilon)$ which is strictly larger than the lefthand side
$|\la 0^{2k} |\psi\ra| \le 1+O(\epsilon)$ if $k\ge 3$ and $\epsilon$ is small enough. 
Meanwhile, Trotter approximation becomes exact only in the limit $\epsilon \to 0$.

\vspace{0.5cm}

\noindent\textbf{Breaking the $r=1$ barrier.}  To unlock the full power of Lee-Yang tensors, such as the analogue of the Perron-Frobenius theorem (Theorem~\ref{thm:degeneracy})
or quasi-polynomial size  state preparation circuits of Section~\ref{sec:state-prep},
one needs sharper lower bounds on the Lee-Yang radius breaking the $r=1$ barrier. 
Let $U=U_\ell \cdots U_2 U_1$ be a Trotter circuit approximating the Gibbs state $\rho \sim e^{-\beta H}$
of some local Hamiltonian $H$. Suppose  we have already proved that each individual gate $U_j$ is a Lee-Yang tensor
with a certain radius $r=r(\ell)>1$, that is, $U_j\in LY(r)$ for all $j$. What can be said about the Lee-Yang radius of the full circuit $U$ ?
The naive application of the tensor contraction lemma (Lemma~\ref{lemma:contraction}) would give $U\in LY(r)$, assuming 
each qubit participates in at least one gate. However, this completely ignores  structure present in the Trotter circuit.
Furthermore, this is not strong enough to break the $r=1$ barrier since
Trotter approximation is exact only in the limit $\ell\to \infty$ and $\lim_{\ell \to \infty} r(\ell)=1$
due to  $\|U_j-I\|=O(1/\ell)$. One might be able to exploit the structure of Trotter circuits by 
contracting some properly chosen sub-circuits of $U$ exactly and applying Lemma~\ref{lemma:contraction} only at the end
to merge the sub-circuits. 
One would need to prove that each sub-circuit is described by a tensor with the Lee-Yang radius
$r\ge 1$ on the internal wires of $U$ and $r>1$ on the external wires of $U$.
Various methods of cutting a large quantum circuit into sub-circuits were investigated in~\cite{peng2019simulating,tang2021cutqc}.

\vspace{0.5cm}
\noindent\textbf{Lee-Yang radius vs spectral gap.}  Theorem~\ref{thm:degeneracy} provides   a link between Lee-Yang radius and spectral properties of 
quantum many-body systems. 
Namely, if the Gibbs state $\rho \sim e^{-\beta H}$ associated with some Hamiltonian $H$
has Lee-Yang radius $r>1$ then $H$ has a positive gap $\delta>0$
between the smallest and the second smallest eigenvalues. Is there a lower bound on the spectral gap $\delta$
in terms of the Lee-Yang gap $r-1$ ? If this is the case, quantum many-body systems breaking the $r=1$ barrier 
might be good targets for quantum simulation algorithms based on the adiabatic evolution approach, as discussed in Section~\ref{subs:EPRgap}.
One may also ask whether the Lee-Yang gap has any implications for rapid mixing of quantum walks used for the Gibbs state preparation,
see e.g.~\cite{chen2023efficient}. A classical analogue of this question has been  recently resolved in~\cite{chen2024spectral}.
This work used the spectral independence method to prove that the existence of a sufficiently large zero-free region for the partition function 
of certain graph-based spin models implies rapid mixing of the Metropolis-Hastings random walk.

\vspace{0.5cm}
\noindent\textbf{Hamiltonian complexity.} 

What is the computational complexity of approximating the ground energy and the partition function 
of  Suzuki-Fisher Hamiltonians?
For concreteness, let us discuss the
complexity landscape for the 
Heisenberg XXZ model.
Given an interaction graph $G=(V,E)$ with $|V|=n$ and constants $d,f\in[-1,1]$, consider the Hamiltonian 
\begin{equation}
\label{app_XXZ_eq1}
    H=-\sum_{(i,j)\in E}w_{i,j}h_{i,j}
\end{equation}
with nonnegative weights $0\leq w_{i,j}\le poly(n)$
and Heisenberg XXZ local terms
\begin{equation}
    h=\frac{1}{2}(f(X\otimes X+Y\otimes Y)+d(I -Z\otimes Z))
    = \left[\ba{cccc}
    0 & 0 & 0 & 0 \\
    0 & d & f & 0 \\
    0 & f & d & 0 \\
    0 & 0 & 0 & 0 \\
    \ea
    \right]. 
\end{equation}
Note that when $d=f=1/2,$ $h$ is the projector onto 
$(\ket{01}+\ket{10})/\sqrt{2}$ and $H$ is unitarily equivalent to the EPR Hamiltonian. 
We consider the complexity of estimating the ground energy  of $H$ to $\eps$-additive error
and estimating the partition function $\mathrm{Tr}(e^{-\beta H})$ to $\epsilon$-multiplicative error,
where $\epsilon \ge poly(1/n)$ and $\beta\le poly(n)$. 
It is well known that the latter problem  is at least as hard as the former (since the free energy at $\beta = O(n\epsilon^{-1})$
provides an $\epsilon$-additive approximation of the ground energy).

Suppose first that $d\leq 0$ and $|f|\leq |d|$. Then $H$ is a Suzuki Fisher Hamiltonian 
since $Z_i Z_j$ terms in $H$ are ferromagnetic and dominating, see Section~\ref{subs:SF}.
Since $H$ is particle number preserving, its ground state can be chosen as $|0^n\ra$,
see Section~\ref{subs:XXZ}. Accordingly, one can compute the ground energy of $H$ classically in time $poly(n)$.
The partition function in the presence of a uniform magnetic field $\tr{(e^{-\beta (H-\mu \sum_{i}Z_i)})}$ can be estimated 
for any constant $\mu>0$
in quasipolynomial time
 using Barvinok's polynomial interpolation method~\cite{barvinok2016book},
as shown by Harrow, Mehraban, and Soleimanifar~\cite{harrow2020classical}. 

Suppose now that $f\ge |d|$. Conjugating $H$ by the Clifford gate $e^{i (\pi/4) X}$
on each qubit gives a Suziki Fisher Hamiltonian 
$H'=-\sum_{(i,j)\in E}w_{i,j}h'_{i,j}$
with the interactions 
\[
h'=\frac{1}{2}(f(X\otimes X+Z\otimes Z)+d(I -Y\otimes Y)).
\]
Clearly, the ground energy and the partition function problems for $H$ and $H'$ are equivalent. 
The Hamiltonian $H'$ is not particle number preserving unless $f=-d$. 
As commented above, $H'$ coincides with the EPR Hamiltonian if $f=d>0$.
Assuming Conjecture~\ref{conj:sgap} of Section~\ref{subs:EPRgap} is true, the ground energy problem for $H'$ (and thus for $H$)  is contained
in the class BQP (for uniform edge weights $w_{i,j}\equiv 1$) since one can use the quantum adiabatic evolution algorithm
to prepare a low-energy state for $H'$, see the discussion at the end of Section~\ref{subs:EPRgap}.

Piddock and Montanaro~\cite{piddock2015complexity} classify the complexity of  ground state problems for a wide variety of $2$-local terms using perturbative gadgets~\cite{kempe2006perturbativegadgets,oliveira20082dlattice}. 
For the Heisenberg XXZ model, they show that if $f<\min\{0,d\},$ then estimating the ground energy  of $H$ to within inverse polynomial additive error is QMA-hard. This implies that approximating the partition function to a given relative error is also QMA-hard. Otherwise (for any values of $d,f$), the Hamiltonian is sign-problem free and the corresponding ground energy problem is contained in StoqMA \cite{bravyi2006merlin}. In this case (a decision version of) the problem of approximating its partition function is in the complexity class AM~\cite{bravyi2006complexity}.

Bravyi and Gosset~\cite{bravyi2016polynomial} show that ferromagnetic XY models admit efficient classical simulation by quantum Monte Carlo methods. This technique yields an efficient algorithm for $d=0$ and $f\geq 0.$ These results are summarized in Figure~\ref{fig:xxzmodels}.

\begin{figure}[h!]\label{fig:xxzmodels}
    \centering
    \includegraphics[width=0.8\linewidth]{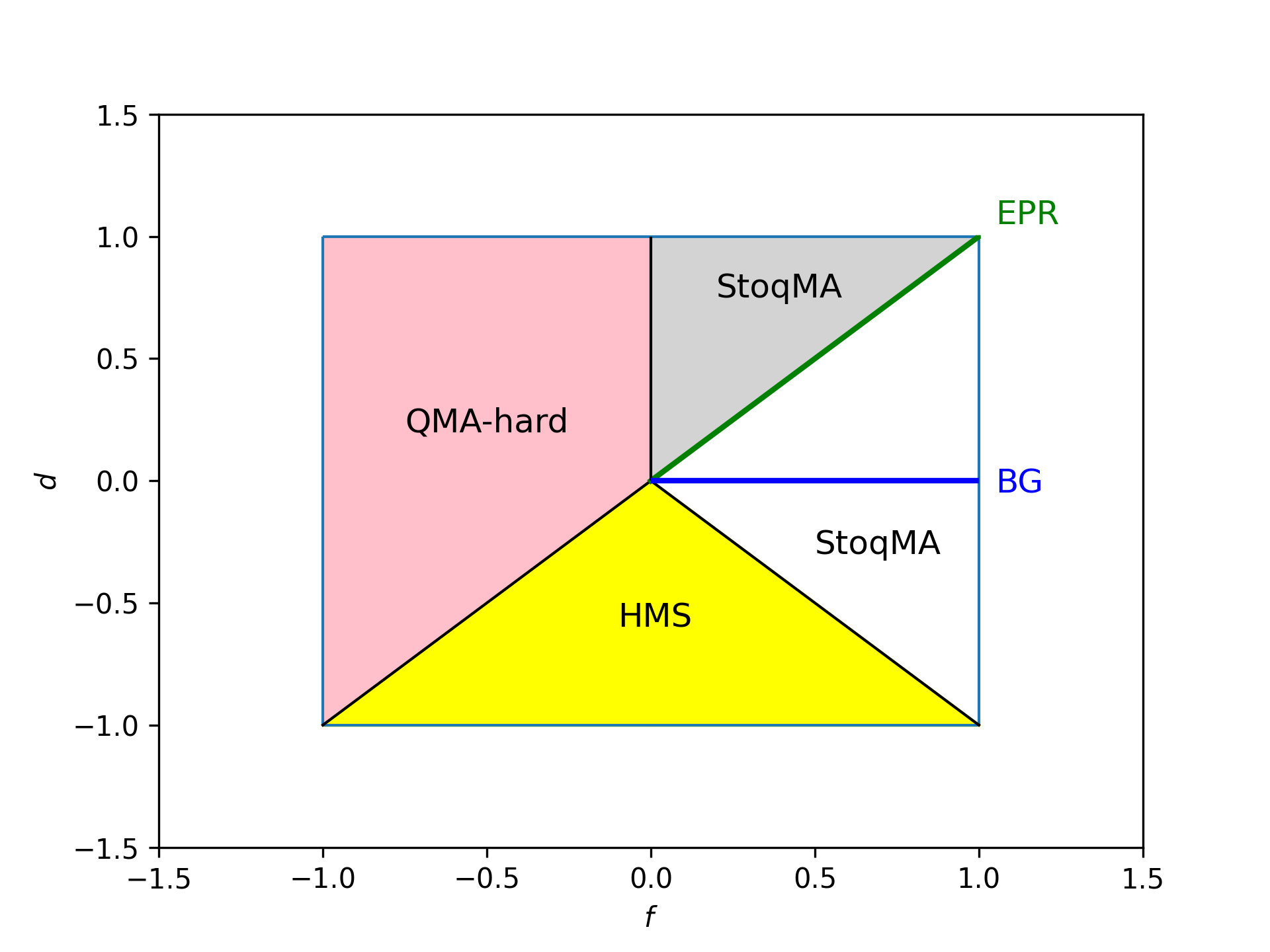}
    \caption{Complexity landscape for XXZ Heisenberg Hamiltonians.
    Ground energy problems for models in the pink region are QMA-hard, while those in the grey and white regions are contained in StoqMA~\cite{piddock2015complexity}. 
Models in the yellow region have the trivial ground state $|0^n\ra$ while the partition function admits 
a quasipolynomial time approximation algorithm in the presence of a uniform magnetic field~\cite{harrow2020classical}. 
The green line represents EPR Hamiltonians.
Assuming Conjecture~\ref{conj:sgap} of Section~\ref{subs:EPRgap} is true, the ground energy problem on the green line is in BQP in the uniformly weighted case ($w_{ij}=1$ for all $i,j$ in Eq.~\eqref{app_XXZ_eq1}). The models on the blue line admit efficient randomized approximation~\cite{bravyi2016polynomial}.
    When $f>0$ the partition functions of these models can be approximated by the partition function of the six-vertex model. On the blue line one obtains a six vertex model in a parameter regime where it admits efficient approximation ~\cite{cai2019approximability}. In contrast, the grey region maps to a  six-vertex model in a parameter regime where it is likely intractable to approximate~\cite{cai2019approximability}. The white region (including the boundaries shown in black and green) maps to a six-vertex model in a parameter regime where its complexity is a mystery.}
\end{figure}

When $f> 0$, the partition function $\tr{(e^{-\beta H})}$ can be approximated by the partition function of the classical $6$-vertex model~\cite{pauling1935ice} on a related edge-ordered graph.  
Assume for simplicity that the edge weights are uniform, $w_{i,j}\equiv 1$.
The Trotter-Suzuki approximation yields $J=\text{poly}(n,\beta,\eps^{-1})$ and $\delta>0$ such that the $Z_J=\tr(G_J\cdots G_2G_1)$, where each $G_i=e^{\delta h}$, is an $\eps$-relative error estimate for the partition function $\tr(e^{-\beta H})$. Introducing a parameter $t=e^{\delta d}\sinh(\delta f)$ one gets

\begin{equation}
e^{\delta h}=\begin{bmatrix}
1 & 0 & 0 & 0\\
0 & t+e^{\delta(d-f)} & t & 0\\
0 & t & t+e^{\delta(d-f)} & 0\\
0 & 0 & 0 & 1
\end{bmatrix}.
\end{equation}
The Trotter circuit $G_J\cdots G_2G_1$ is comprised entirely of gates which act non-trivially on exactly $2$ qubits, hence the graph $\Gamma$ representing the tensor network description of $\tr(G_J\cdots G_2G_1)$ is $4$-regular. The $6$-vertex model with parameters $a,b,c\geq 0$ is defined on $4$-regular graphs where each vertex has an ordering of its incident edges. The states of the six-vertex model are the Eulerian orientations of the underlying graph denoted $\calE$. Given an Eulerian orientation $\tau$ each vertex can be assigned a type based on the direction of the incident edges as shown in Figure~\ref{fig:six-vertex}. Then for each $\tau\in\calE$ assign a weight $w(\tau)=a^{n_1}b^{n_2}c^{n_3}$ where $n_i$ denotes the number of vertices of type $i$ in $\tau$.
\begin{figure}[h] \label{fig:six-vertex}
    \centering
    \includegraphics[width=0.65\linewidth]{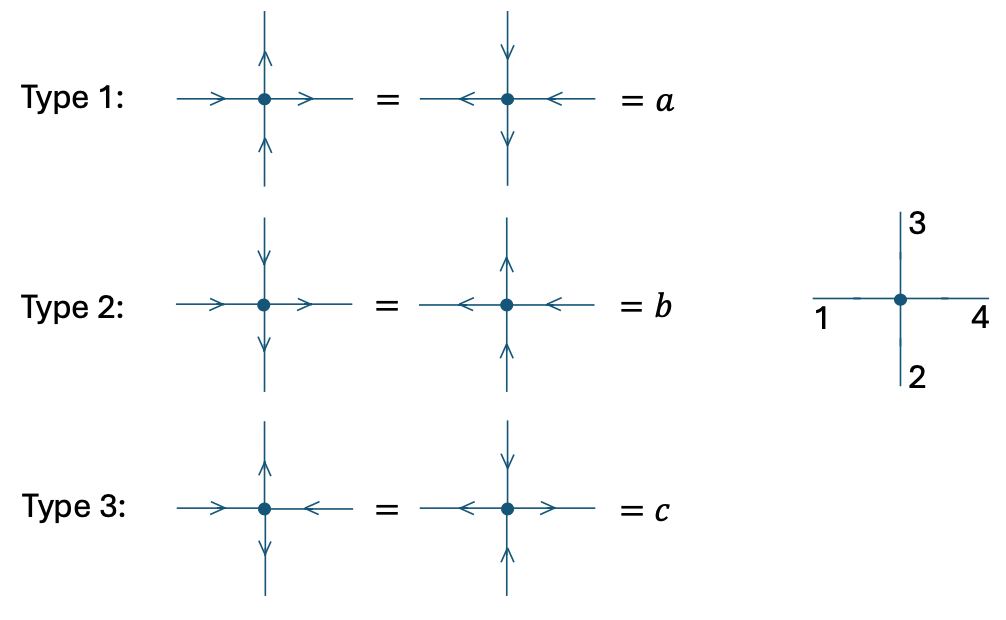}
    \caption{Each vertex is assigned a type $1,2,$ or $3$ depending the orientation of its incident edges and a fixed ordering (shown on the right).}
\end{figure}
Then the partition function is
\begin{equation}
    Z(\Gamma,a,b,c)=\sum_{\tau\in \calE}a^{n_1}b^{n_2}c^{n_3}.
\end{equation}
A bit string $x\in\bits{|E(\Gamma)|}$ defines an Eulerian orientation of $\Gamma$ such that an edge $e\in E(\Gamma)$ is oriented forward in time if $x_e=0$ and backward in time if $x_e=1.$ Then there is a one-to-one correspondence between the three types of nonzero matrix elements in $e^{\delta h}$ and the three vertex types
of the six-vertex model. Hence we have agreement of the partition functions
\begin{align}
    Z_J(G_J\ldots G_2G_1)&=\sum_{x\in\bits{|E(\Gamma)|}}\prod_{u\in V(\Gamma)}\bra{x_{e_1(u)}x_{e_2(u)}}G_u\ket{x_{e_3(u)x_{e_4(u)}}}\\
    &=\sum_{\tau\in\calE}1^{n_1}t^{n_2}(t+e^{\delta(d-f)})^{n_3}=Z(\Gamma,1,t,t+e^{\delta(d-f)})
\end{align}
where $e_i(u)$ is edge $i$ incident to $u$ as determined by the ordering on $\Gamma.$
Cai, Lu, and Liu~\cite{cai2019approximability} showed that when $a^2\leq b^2+c^2$, $b^2\leq a^2+c^2$, and $c^2\leq a^2+b^2$, the partition function 
$Z(\Gamma,a,b,c)$
can be estimated efficiently using the Markov chain Monte Carlo method for any $4$-regular graph $\Gamma$.
They also showed that when $a>b+c$ or $b>a+c$ or $c>a+b$ the six-vertex model cannot be approximated efficiently unless $\mathrm{RP}=\mathrm{NP}$, which is believed to be extremely unlikely.

This means that when $d=0$ and $f\geq0,$ the corresponding XXZ Heisenberg models can be simulated classically. In this case, Cai et. al recover the classical simulation of this model demonstrated by Bravyi and Gosset~\cite{bravyi2016polynomial}. On the other hand, when $d>f>0$, we are in a parameter regime where $c>a+b$ and the six-vertex model is likely intractable to approximate.

\section{Acknowledgments}
  Research at Perimeter Institute is supported by the Government of Canada through the Department of Innovation, Science, and Economic Development, and by the Province of Ontario through the Ministry of Colleges and Universities. DG is a fellow of the Canadian Institute for Advanced Research, in the quantum information science program. DG and YL acknowledge support from IBM Research. DG, YL, and BW acknowledge the support of the Natural Sciences and Engineering Research Council of Canada.

\appendix 

\section{Proof of Lemma~\ref{lemma:contraction}}
\label{app:A}

Assume wlog that $\phi$ is obtained from $\psi$ by contracting the first two indices  (otherwise, permute the indices). 
 Write
 $z=(z_1,z_2,z')\in \CC^n$, where $z'\in \CC^{n-2}$.  The generating polynomials of $\psi$ and $\phi$ can be written as
\be
\label{contraction_eq0}
f_\psi(z) = g_{00}(z') + z_1 g_{10}(z') + z_2 g_{01}(z') + z_1 z_2 g_{11}(z') \quad \mbox{and} \quad f_\phi(z') = g_{00}(z') + g_{11}(z')
\ee
for some  polynomials $g_{ab}\, : \, \CC^{n-2}\to \CC$.

Suppose $r_1r_2>1$.
Let us assume that  $f_\phi\notin LY_{n-2}(\rr')$ and show that this assumption leads to a contradiction.
Indeed, if $f_\phi\notin LY_{n-2}(\rr')$ then  $f_\phi(z')=0$ for some $z'\in \odisk{\rr'}$.
By a slight abuse of notations, let us identify the polynomials $g_{ab}$ and their evaluations
at the point $z'$. Then Eq.~(\ref{contraction_eq0}) gives
\be
\label{contraction_eq0'}
g_{00}+g_{11}=0.
\ee
Condition  $\psi \in LY_n(\rr)$ implies $f_\psi(z)\ne 0$ for $z=(z_1,z_2,z')$ with $z_1\in \odisk{r_1}$, $z_2\in \odisk{r_2}$.
From Eqs.~(\ref{contraction_eq0},\ref{contraction_eq0'}) one gets
\be
\label{contraction_eq1}
(1-z_1z_2)g_{00} + z_1g_{10} + z_2g_{01} \ne 0 \quad \mbox{for all $z_1\in \odisk{r_1}$ and $z_2\in \odisk{r_2}$}.
\ee
We shall need  the following lemma due to Asano~\cite{asano1970theorems} and Ruelle~\cite{ruelle2010characterization}.
\begin{lemma}
\label{lemma:Asano}
Let $K_1$ and $K_2$ be any closed subsets of $\CC$ such that $0\notin K_1\cup K_2$.
Suppose $a,b,c,d$ are complex numbers such that 
$a + bz_1 + cz_2 + dz_1z_2\ne 0$
whenever $z_1\notin K_1$ and $z_2\notin K_2$. Then $a+dz\ne 0$ whenever $z\notin -K_1 K_2$.
\end{lemma}
\noindent
Here $-K_1 K_2=\{-z_1z_2\, : \, z_1\in K_1,\; z_2\in K_2\}$.
A proof of Lemma~\ref{lemma:Asano} can be found in~\cite{ruelle2010characterization}.
Choose $K_j$ as the complement of $\odisk{r_j}$ with $j=1,2$. Since $r_j>0$, one has 
$0\notin K_1\cup K_2$. Choose $a=g_{00}$, $b=g_{10}$, $c=g_{01}$, $d=-a$.
Then Eq.~(\ref{contraction_eq1}) is equivalent to $a + bz_1 + cz_2 + dz_1z_2\ne 0$
for $z_1\notin K_1$ and $z_2\notin K_2$.  Lemma~\ref{lemma:Asano} then implies
$a(1-z)\ne 0$ for $z\notin -K_1 K_2$. Obviously, $-K_1K_2=K_1 K_2$ is the complement of $\odisk{r_1 r_2}$.
From $r_1r_2>1$ one gets $1\in \odisk{r_1 r_2}$ and thus  $1\notin -K_1 K_2$ arriving at a contradiction since $a(1-z)=0$ for $z=1$.
This proves $f_\phi\in LY_{n-2}(\rr')$.

Suppose now that $r_1r_2=1$. Choose an arbitrary sequence $\{s_k\}_{k=1}^\infty$ such that $s_k\in [0,1)$ for all $k$
and $\lim_{k\to \infty} s_k=1$.
Consider a sequence of tensors 
\[
|\psi_k\ra = (|0\ra\la 0| + s_k |1\ra\la 1|)^{\otimes n} |\psi\ra.
\]
It follows directly from the definitions that 
$\psi_k \in LY_n(\rr/s_k)$ and $\lim_{k\to \infty} \psi_k=\psi$.
 Let $\phi_k$ be the tensor obtained from $\psi_k$ by contracting the first two indices.
Clearly, $\lim_{k\to \infty} \phi_k=\phi$. Since $(r_1/s_k)(r_2/s_k)>1$, we can apply the first part of the lemma
to get $\phi_k\in LY_{n-2}(\rr'/s_k)$. Since $s_k\in (0,1)$ we also have $\phi_k\in LY_{n-2}(\rr')$ for all $k\geq 1$. From Corollary \ref{corol:seq}
 we infer that either $\phi$ is identically zero, or $\phi\in LY_{n-2}(\rr')$.

\section{Proof of Theorem \ref{thm:sf71}}
\label{app:sf71}

For completeness we include a proof of Theorem \ref{thm:sf71}. We follow the strategy from Ref.~\cite{sf71}.
We shall first establish the result in the special case where Eq.~\eqref{eq:abs} holds with a strict inequality, i.e.
\begin{equation}
|\langle 0^{m}|\psi\rangle|> \frac{1}{4}\sum_{y\in \{0,1\}^{m}} |\langle y|\psi\rangle|.
\label{eq:abs_strict}
\end{equation}
The case where equality holds in Eq.~\eqref{eq:abs_strict} then follows via a continuity argument based on Hurwitz's theorem. In particular, if equality holds in Eq.~\eqref{eq:abs_strict} for $\psi$,  we can choose a sequence $\{\psi_j\}_{j\geq 1}$ defined by $|\psi_j\rangle=|\psi\rangle+2^{-j} (|0^m\rangle\pm |1^m\rangle)$. The corresponding vectors all satisfy Eqs. (\ref{eq:sf1},\ref{eq:abs_strict}) and $\psi_j \rightarrow \psi$ entrywise. By Corollary \ref{corol:seq}, $\psi\in LY(1)$.

Firstly note that using the fact that $X^{\otimes m}|\psi\rangle=\pm |\psi^*\rangle$ we have
\[
\sum_{y\in \{0,1\}^{m-1}, y\neq 0^m} |\langle 0 y|\psi\rangle|=\sum_{y\in \{0,1\}^{m-1}, y\neq 1^m} |\langle 1 y|\psi\rangle|\quad \text{and}\quad  |\langle 0^m|\psi\rangle|=|\langle 1^m|\psi\rangle|.
\]
Plugging into Eq.~\eqref{eq:abs_strict} we infer
\begin{equation}
|\langle 0^{m}|\psi\rangle|> \sum_{y\in \{0,1\}^{m-1}:y\neq 0^m} |\langle 0 y|\psi\rangle|.
\label{eq:abs3}
\end{equation}

Write 
\begin{equation}
f_{\psi}(s_1,\ldots, s_m)=\sum_{x\in \{0,1\}^m} \langle x|\psi\rangle s_1^{x_1}\ldots s_m^{x_m}= s_1 A_1(s_2,\ldots, s_m)+A_0(s_2,\ldots, s_m).
\label{eq:ppsi}
\end{equation}
We first show that 
\begin{equation}
A_0(s_2,\ldots, s_m)\neq 0 \quad \text{ whenever } \quad |s_2|, |s_3|,\ldots, |s_m| \leq 1.
\label{eq:a0}
\end{equation}
To see this note that 
\begin{align}
|A_0(s_2,\ldots, s_m)|&=\big|\sum_{x\in \{0,1\}^{m-1}} \langle 0 x|\psi\rangle s_2^{x_2} s_3^{x_3}\ldots s_m^{x_m}\big|\\
&\geq |\langle 0^{m}|\psi\rangle|- \sum_{y\in \{0,1\}^{m-1}, y\neq 0^m} |\langle 0 y|\psi\rangle|\\
&> 0
\end{align}
where in the first inequality we used the fact that $|s_j|\leq 1$ for $j\in \{2,3,\ldots, m\}$ and in the second inequality we used Eq.~\eqref{eq:abs3}.

Now suppose that $(t_1, t_2,\ldots, t_m)$ is a root of $f_{\psi}$ and that
\begin{equation}
|t_j| < 1 \quad j\in \{2,3,\ldots, m\}
\label{eq:sleq}
\end{equation}
To complete the proof of the theorem it suffices to show that $|t_1|\geq 1$. Firstly, note that 
\begin{equation}
A_1(t_2,t_3,\ldots, t_m)\neq 0.
\label{eq:a1}
\end{equation}
To see this, note that otherwise we would have $f_{\psi}(t_1,t_2,\ldots, t_m)=0=A_0(t_2,\ldots, t_m)$ which contradicts Eq.~\eqref{eq:a0}.  Using Eqs.~(\ref{eq:ppsi},\ref{eq:a1}) and the fact that $f_{\psi}(t_1,t_2,\ldots, t_m)=0$ gives
\[
|t_1|=\left|\frac{A_0(t_2,t_3,\ldots, t_m)}{A_1(t_2,t_3,\ldots, t_m)}\right|\geq \left(\max_{|s_2|,|s_3|,\ldots, |s_m|\leq 1 }\left|\frac{A_1(s_2,s_3,\ldots, s_m)}{A_0(s_2,t_3,\ldots, s_m)}\right| \right)^{-1}=1,
\]
where the last inequality follows from the lemma below.
\begin{lemma}
\begin{equation}
\max_{|s_2|,|s_3|,\ldots, |s_m|\leq 1 }\left|\frac{A_1(s_2,s_3,\ldots, s_m)}{A_0(s_2,t_3,\ldots, s_m)}\right|=1
\label{eq:smax}
\end{equation}
\end{lemma}
\begin{proof}
Suppose that $(r_2,r_3,\ldots, r_m)$ maximizes the LHS of Eq.~\eqref{eq:smax} and such that $|r_j|<1$ for some $j$. In this case one can choose $w_j\in \mathbb{C}$ such that $|w_j|=1$ and such that the tuple
\[
(r_2,\ldots, r_{j-1},w_j,r_{j+1},\ldots, r_m)
\]
maximizes the LHS of Eq.~\eqref{eq:smax}. This follows by using the maximum modulus principle from complex analysis, on the univariate function
\[
F(w_j)=\frac{A_1(r_2,\ldots, r_{j-1},w_j, r_{j+1},\ldots, r_m)}{A_0(r_2,\ldots, r_{j-1},w_j, r_{j+1},\ldots, r_m)}
\]
from which we infer that the maximum of $|F(w_j)|$ over the unit circle is achieved on its boundary. Applying this fact recursively shows that the maximum is achieved by a vector with all coordinates on the unit circle:
\[
\max_{|s_2|,|s_3|,\ldots, |s_m|\leq 1 }\left|\frac{A_1(s_2,s_3,\ldots, s_m)}{A_0(s_2,s_3,\ldots, s_m)}\right|=\max_{|s_2|,|s_3|,\ldots, |s_m|=1 }\left|\frac{A_1(s_2,s_3,\ldots, s_m)}{A_0(s_2,s_3,\ldots, s_m)}\right|.
\]
Finally, using the symmetry $X^{\otimes m}|\psi\rangle=\pm |\psi^*\rangle$ and Eq.~\eqref{eq:ppsi} we see that
\[
\left(\prod_{j=1}^{m} s_j\right)f_{\psi^{*}}(1/s_1,\ldots, 1/s_m)=\pm f_{\psi}(s_1,\ldots, s_m)
\]
and therefore
\[
\left(\prod_{j=2}^{m} s_j\right) \left(A_1(s_2,\ldots, s_m)\right)^*=\pm A_0(s_2,\ldots, s_m)  \qquad \text{whenever} \quad |s_2|=|s_3|=\ldots =|s_m|=1.
\]
This gives
\[
\max_{|s_2|,|s_3|,\ldots, |s_m|=1 }\left|\frac{A_1(s_2,s_3,\ldots, s_m)}{A_0(s_2,s_3,\ldots, s_m)}\right|=\max_{|s_2|,|s_3|,\ldots, |s_m|=1 }\left|\pm \frac{\left(A_0(s_2,s_3,\ldots, s_m)\right)^*}{A_0(s_2,s_3,\ldots, s_m)}\right|=1.
\]
\end{proof}

\bibliographystyle{unsrt}
\bibliography{mybib}

\end{document}